\documentclass[final,10pt,twocolumn]{autart} 

\usepackage{bbold}
\usepackage{mathbbol}
\usepackage{amssymb}

\usepackage{graphicx}
\usepackage[table,pdftex,hyperref,svgnames]{xcolor}
\usepackage{url,doi,hyperref}
\hypersetup{colorlinks=true, linkcolor=blue, breaklinks=true, urlcolor=blue}

\usepackage{flushend}
\usepackage{blindtext, graphicx}
\usepackage[utf8]{inputenc}
\usepackage{enumerate}
\usepackage{enumitem}
\usepackage{tabulary}
\usepackage{booktabs}

\graphicspath{{./images/}{./fig}}

\usepackage[cmex10]{amsmath}

\let\theoremstyle\relax

\usepackage{amsthm}

\usepackage{amsfonts}
\usepackage{bm}

\usepackage[caption=false,font=footnotesize]{subfig}

\usepackage{lipsum}
\usepackage{mathtools}
\usepackage{cuted}

\usepackage{mathtools,bbm}

\usepackage[prependcaption,colorinlistoftodos]{todonotes}

\newcommand\oprocendsymbol{\hbox{$\square$}}
\newcommand\oprocend{\relax\ifmmode\else\unskip\hfill\fi\oprocendsymbol}

\newtheorem{theorem}{Theorem}[section]
\newtheorem{proposition}[theorem]{Proposition}

\newtheorem{definition}[theorem]{Definition}
\newtheorem{lemma}[theorem]{Lemma}

\newtheorem{remark}[theorem]{Remark}

\newtheorem{facts}[theorem]{Fact}

\DeclareSymbolFont{bbold}{U}{bbold}{m}{n}
\DeclareSymbolFontAlphabet{\mathbbold}{bbold}
\newcommand{\vect}[1]{\mathbbold{#1}}

\usepackage{mathabx}
        \makeatletter
        \newcommand{\ovset}[3][0ex]{%
          \mathrel{\mathop{#3}\limits^{
            \vbox to#1{\kern0\ex@
            \hbox{$\scriptstyle#2$}\vss}}}}
        \makeatother

\newcommand{\C}{\mathcal{C}}
\newcommand{\dsG}{\ovset{\rightharpoonup}{\mathcal{G}}}

\newcommand{\uG}{|\mathcal{G}|}
\newcommand{\dE}{\ovset{\rightharpoonup}{E}}

\usepackage{enumitem}
\renewcommand{\theenumi}{(\roman{enumi})}

\graphicspath{ {fig/} }

\newenvironment{example}{\textit{Example:}}

\renewcommand{\baselinestretch}{1}

\DeclareUnicodeCharacter{2010}{-} 

\begin{document}
\renewcommand{\baselinestretch}{0.93}
\setlength\parindent{1em}

\begin{frontmatter}

\title{Structural Balance and Interpersonal Appraisals Dynamics: Beyond All-to-All and Two-Faction Networks}
\thanks[footnoteinfo]{This material is
  based upon work supported in part by the ETH Zurich funding.}
  
\author[IfA]{Wenjun Mei}\ead{wmei@ethz.ch}, \author[cas]{Ge Chen}\ead{chenge@amss.ac.cn}, \author[UCSB]{Noah E. Friedkin}\ead{friedkin@soc.ucsb.edu}, \author[IfA]{Florian D\"{o}rfler}\ead{dorfler@ethz.ch}

\address[IfA]{Automatic Control Laboratory, ETH Zurich, Switzerland}
\address[cas]{Academy of Mathematics and Systems Science, Chinese Academy of Sciences, Beijing, China}
\address[UCSB]{Department of Sociology, University of California, Santa Barbara}


\begin{abstract}
Structural balance theory describes stable configurations of topologies of signed interpersonal appraisal networks. Existing models explaining the convergence of appraisal networks to structural balance either diverge in finite time, or could get stuck in jammed states, or converge to only complete graphs. In this paper we study the open problem how steady non-all-to-all structural balance emerges via local dynamics of interpersonal appraisals. We first compare two well-justified definitions of structural balance for general non-all-to-all graphs, i.e., the triad-wise structural balance and the two-faction structural balance, and thoroughly study their relations. Secondly, based on three widely adopted sociological mechanisms: the symmetry mechanism, the influence mechanism, and the homophily mechanism, we propose two simple models of gossip-like appraisal dynamics, the symmetry-influence-homophily (SIH) dynamics and the symmetry-influence-opinion-homophily (SIOH) dynamics. In these models, the appraisal network starting from any initial condition almost surely achieves non-all-to-all triad-wise and two-faction structural balance in finite time respectively. Moreover, the SIOH dynamics capture the co-evolution of interpersonal appraisals and individuals' opinions. Regarding the theoretical contributions, we show that the equilibrium set of the SIH (SIOH resp.) dynamics corresponds to the set of all the possible triad-wise (two-faction resp.) structural balance configurations of the appraisal networks. Moreover, we prove that, for any initial condition, the appraisal networks in the SIH (SIOH resp.) dynamics almost surely achieve triad-wise (two-faction resp.) structural balance in finite time. Numerical studies of the SIH dynamics are conducted on how the final proportion of negative links depends on the initial proportion of negative links and the network topology. These simulation results imply some insightful take-home messages on whether multilateral relations reduce or exacerbate conflicts. 
\end{abstract}

\end{frontmatter}

\section{Introduction}

\paragraph*{Motivation and problem description}

Properties and dynamics on/of signed networks have been widely studied in social science as well as applied mathematics. \emph{Structural balance} (also referred to as \emph{social balance}) theory, first proposed in the seminal works by Heider~\cite{FH:44,FH:46}, characterizes the stable configurations of interpersonal appraisal networks with both friendly and antagonistic relations. An appraisal network satisfies structural balance if each individual obeys the famous Heider's axioms:  Friends' friends are friends; Friends' enemies are enemies; Enemies' friends are enemies; Enemies enemies are friends.''  As suggested by~\cite{FH:44,FH:46}, imbalance of interpersonal relations sensed by individuals leads to cognitive dissonances that the individuals strive to resolve.
Dynamic structural balance theory, aiming to explain how an initially unbalanced network evolves to a balanced state, has recently attracted much interest. In existing models~\cite{TA-PLK-SR:05,TA-PLK-SR:06,AvdR:11,MM-MF-PJK-HRR-MAS:11,KK-PG-PG:05,SAM-JK-RDK-SHS:11,VAT-PVD-PDL:13,SW-MC-KK-AB:15,WM-PCV-GC-NEF-FB:17f}, appraisal networks either diverge in finite time, or get stuck in unbalanced equilibria, or converge to all-to-all graphs satisfying structural balance. It remains an open problem what models lead to the convergence of appraisal networks to structural
balance with arbitrary non-all-to-all network topologies. In this paper, we address this open problem. To be more specific, we propose and analyze dynamics of interpersonal appraisals in which the appraisal networks are initially non-all-to-all and converge to structurally balanced equilibria with also a non-all-to-all topology.

Before studying the convergence of interpersonal appraisals to non-all-to-all structural balance, one would need to first clarify what structural balance means in this scenario. In this paper, the interpersonal appraisal networks are modeled as signed graphs. Intuitively, structural balance in non-all-to-all signed graphs can be defined in two ways. The first one is a straightforward generalization of the aforementioned Heider's four axioms. Mathematically, Heider's four axioms mean that, every three nodes (i.e. individuals) in the appraisal network form a ``positive triad'', i.e., a triad with either three symmetric positive relations or two symmetric positive relations and one symmetric negative relation. Since every three nodes in a non-all-to-all graph do not necessarily form a triad, an intuitive generalization of Heider's axioms is to require that every existing triad in the appraisal network is positive. We refer to such a definition as \emph{triad-wise structural balance}. Compared to Heider's classic structural balance, triad-wise balance allows for more realistic scenarios, e.g., one does not necessarily know all of her/his friends' friends. The second definition, referred to as \emph{two-faction structural balance}, was first proposed by Cartwright et al.~\cite{DC-FH:56} and has been widely adopted in the studies of opinion dynamics on signed networks~\cite{CA:13,JL-XC-TB-MAB:17,GS-CA-JSB:19}. It requires that the appraisal network can be partitioned into two antagonistic factions where the inter-faction relations are all non-positive and the relations within each faction are all non-negative. 

In this paper, we first thoroughly study the relations between triad-wise structural balance and two-faction structural balance, and then propose two simple discrete-time dynamic models, the \emph{symmetry-influence-homophily} (SIH) model and the \emph{symmetry-influence-opinion-homophily} (SIOH) model, such that the interpersonal appraisals almost surely converge to triad-wise and two-faction non-all-to-all structural balance respectively. In our models, links in appraisal networks take values from $\{-1,1,0\}$, corresponding to antagonistic, friendly, and no/neutral relations respectively. At any time, one such link is activated and updated via (some of) the following sociological mechanisms of local interactions: the symmetry mechanism~\cite{RME:76}, the influence mechanism~\cite{NEF-ECJ:11}, the person-person homophily mechanism~\cite{PFL-RKM:54}, and the person-opinion homophily~\cite{DC-FH:56}. The symmetry mechanism means that individuals tend to be friendly to (antagonistic against resp.) those who are friendly to (antagonistic against resp.) themselves. The influence mechanism assumes that one individual's attitude towards another is influenced by their mutual social neighbors' attitudes. The homophily mechanisms in general mean that individuals are friendly to (antagonistic against resp.) each other if they hold similar opinions on some issues, or similar appraisals of others.

\paragraph*{Literature review}

Following the early works by Heider~\cite{FH:44,FH:46}, static structural balance theory has been extensively studied, including the characterization of the balanced configurations~\cite{DC-FH:56,WdN:99}, the degree of balance~\cite{DC-TCG:66,NMH-RBH-CBDS:69}, and empirical validations~\cite{HBM-RR:85,CMR-NEF:17,GF-GI-CA:11}. Generalized structural balance has also been studied in terms of removing some of Heider's four axioms, e.g., see~\cite{JAD:67,OA-JNL-FB-NEF-AKS-BU:17v}. Different from these generalizations, the triad-wise structural balance proposed in this paper extends the applicability of the classic Heider's structure balance from all-to-all graphs to arbitrary topologies by requiring all the existing triads in the appraisal network to satisfy Heider's four axioms. For a comprehensive review of static structural balance theory, we refer to~\cite{XZ-DZ-FYW:15}.

Previous works on dynamic structural balance theory include the discrete-time \emph{local triad dynamics} (LTD)~\cite{TA-PLK-SR:05} and \emph{constrained triad dynamics}~\cite{TA-PLK-SR:06}. These models suffer from the existence of unbalanced equilibria, i.e., the \emph{jammed states}. Other models based on network games are proposed by Malekzadeh et al.~\cite{MM-MF-PJK-HRR-MAS:11} and van de Rijt~\cite{AvdR:11}. The model in~\cite{MM-MF-PJK-HRR-MAS:11} applies to only complete graphs. The model in~\cite{AvdR:11} propose smoothed best-response dynamics with vanishing noise, in which some generalized structurally balanced graphs with arbitrary topologies are the only stochastically stable configurations. But stochastic stability is a weaker condition than almost-sure convergence, and the model in~\cite{AvdR:11} actually does not admit any steady state. Friedkin et al.~\cite{NEF-AVP-FB:18m} propose a generalized model based on rule of transitive closure in which the temporal elimination of violations of Heider's four axioms appears as a special case.

Dynamic structural balance models with continuous link weights have also been widely-studied, e.g., see~\cite{KK-PG-PG:05,SAM-JK-RDK-SHS:11,VAT-PVD-PDL:13,WM-PCV-GC-NEF-FB:17f,PCV-FB:19g}. In terms of the microscopic sociological mechanism, these models are all based on either the influence mechanism~\cite{NEF-ECJ:11} or the homophily mechanism~\cite{PFL-RKM:54}. In these models, the appraisal networks either diverge in finite time~\cite{KK-PG-PG:05,SAM-JK-RDK-SHS:11} or converge from certain sets of initial conditions to complete graphs satisfying structural balance~\cite{WM-PCV-GC-NEF-FB:17f,PCV-FB:19g}.

\paragraph*{Contributions}

The contributions of this paper are as follows. Firstly, we conduct comprehensive analysis of the relation between triad-wise structural balance and two-faction structural balance for non-all-to-all appraisal networks. We show that if a directed signed graph satisfies triad-wise structural balance, then any node's ego-network satisfies two-faction structural balance. Moreover, we provide a graph-theoretic condition with clear geometric interpretation, under which these two definitions of structural balance are equivalent. 

Secondly, to the best of our knowledge, the SIH and SIOH models proposed in this paper are the first to establish the almost-sure convergence of interpersonal appraisals, via local social interactions, to triad-wise and two-faction non-all-to-all structural balance respectively. Since real social networks are usually not all-to-all graphs, our models significantly broaden the applicability of the dynamic structural balance theory and, in this sense, are more realistic than previous models.  

Thirdly, regarding the theoretical analysis of the SIH (SIOH resp.) model, we prove that its equilibrium set corresponds to the set of all the possible triad-wise (two-faction resp.) balanced configurations. We further show that, for any initial condition, the appraisal network almost surely reaches a triad-wise (two-faction resp.) balanced equilibrium in finite time. For the SIOH dynamics, we also prove the convergence of individuals' opinions. 

Fourthly, by numerical studies, we investigate how the final degree of conflicts of the SIH dynamics, i.e., the ratio of negative links at the final steady states, depends on the initial degree of conflicts and the network topology. Simulation results lead to very clear and insightful sociological interpretations: 1) The final degree of conflicts tend to increase with the initial degree of conflicts and such correlation is stronger in sparser networks; 2) Multilateral relations tend to reduce the final degree of conflicts when the initial degree of conflicts is high, but tend to increase the final degree of conflicts when the initial degree of conflicts is low. In addition, such correlations are stronger in sparser networks.

\paragraph*{Organization}
The rest of this paper is organized as follows. Section 2 is about some basic notations and definitions. Section 3 presents the results on the relation between triad-wise structural balance and two-faction structural balance; In Section 4, we propose and analyze the SIH and SIOH dynamics respectively. Section 5 presents the numerical study and some further discussions. All the proofs are provided in the Appendix.

\section{Notations and Definitions}

Let $\phi$ be the empty set and $\mathbb{N}$ be the set of natural numbers. Let $V=\{1,2,\dots,n\}$. 

Two type of graphs are involved in this paper. A \emph{directed and unweighted signed graph} with $n$ nodes is denoted by $\dsG=(V,\dE^+,\dE^-)$, where $V$ is the nodes set and the two disjoint sets $\dE^+\subseteq V\times V$ and $\dE^-\subseteq V\times V$ are the set of all the positive and negative directed links in $\dsG$ respectively. We only consider unweighted graphs and thereby the term ``unweighted'' is omitted in the rest of this paper. An \emph{undirected unsigned graph} is denoted by $\uG=(V,E)$, where $E$ is the set of all the undirected links. A \emph{directed path} (\emph{path} resp.) in $\dsG$ ($\uG$ resp.) from node $i_1$ to node $i_m$ is an ordered sequence of nodes $i_1,\,\dots,i_m$ such that $(i_k,i_{k+1})\in \dE^+\cup \dE^-$ ($\{i_k,i_{k+1}\}\in E$ resp.) for any $k=1,\dots, m-1$. This directed path (path resp.) has length $m-1$. A graph $\dsG$ ($\uG$ resp.) is \emph{strongly connected} (\emph{connected} resp.) if, for any $i,\, j\in V$, there exists at least one directed path (path resp.) from $i$ to $j$. An ordered sequence $(i_1,\dots, i_m)$ of non-repeating nodes is a cycle with length $m$ in $\dsG$ ($\uG$ resp.) if $(i_k,i_{k+1})\in \dE^+\cup \dE^-$ ($\{i_k,i_{k+1}\}\in E$ resp.) for any $k\in \{1,\dots,m\}$. Here we take $i_{m+1}$ as $i_1$. A cycle on $\dsG$ is \emph{positive} if it contains either zero or an even number of negative links. A triad is a cycle with length 3 and a pair of nodes linking to each other is a cycle with length 2. 

Given a graph $\dsG$ ($\uG$ resp.) and a subset of nodes $\tilde{V}\subseteq V$, denote by $\dsG_{\tilde{V}}$ ($\uG_{\tilde{V}}$ resp.) the \emph{induced subgraph} of $\dsG$ ($\uG$ resp.) associated with $\tilde{V}$. Namely, $\dsG_{\tilde{V}}=(\tilde{V},\dE^+_{\tilde{V}},\dE^-_{\tilde{V}})$, where $\dE^+_{\tilde{V}} = (\tilde{V}\times \tilde{V}) \cap \dE^+$ and $\dE^-_{\tilde{V}}  = (\tilde{V}\times \tilde{V}) \cap \dE^-$, and $\uG_{\tilde{V}}=(\tilde{V},E_{\tilde{V}})$, where $E_{\tilde{V}} = \big\{ \{i,j\} \big|i,\,j\in \tilde{V} \big\}\cap E$.
Given a subset of nodes $\tilde{V}\subseteq V$ and subsets of links $\tilde{\dE}^+\subseteq \dE^+_{\tilde{V}}$ and $\tilde{\dE}^-\subseteq \dE^-_{\tilde{V}}$ ($\tilde{E}\subseteq E_{\tilde{V}}$ resp.), the graph $\tilde{\dsG}=(\tilde{V},\tilde{\dE}^+,\tilde{\dE}^-)$ ($\tilde{\uG}=(\tilde{V},\tilde{E})$ resp.) is called a \emph{subgraph} of $\dsG$ ($\uG$ resp.) with the nodes set $\tilde{V}$.

Given a group of $n$ individuals, their interpersonal appraisals are characterized by a ternary \emph{appraisal matrix} $X=(X_{ij})_{n\times n}\in \{-1,0,1\}^{n\times n}$. Here $X_{ij}=1$ ($X_{ij}=-1$ resp.) means that individual $i$ is friendly (antagonistic resp.) to individual $j$, while $X_{ij}=0$ means that either $i$ does not know $j$ or $i$ holds neutral attitude towards $j$. An appraisal matrix $X$ induces a directed and signed graph $\dsG(X)=\big(V,\dE^+(X),\dE^-(X)\big)$, referred to as the \emph{appraisal network}, where $\dE^+(X)=\{(i,j)\in V\times V\,|\, X_{ij}=1\}$ and $\dE^-(X)=\{(i,j)\in V\times V\,|\, X_{ij}=-1\}$. The matrix $X$ also induces an undirected unsigned graph $\uG(X)=(V,E(X))$, where $E(X)=\big{\{}\{i,j\}\,|\,i,j\in V,\,\, X_{ij}\neq 0\big{\}}$. In this paper, we use the terms ``graph'' and ``network'' interchangeably. The appraisals network $\dsG(X)$ is \emph{bilateral}, or, equivalently, the appraisal matrix $X$ is bilateral, if ``$X_{ij}\neq 0$ $\Leftrightarrow$ $X_{ji}\neq 0$'' holds for any $i,\, j\in V$. We consider no self loop, i.e., $X_{ii}=0$ for any $i\in V$. For any  $X\in \{-1,0,1\}^{n\times n}$, define $\uG(X)=(V,E)$ as $V=\{1,\dots, n\}$ and $E=\big{\{} \{i,j\}\big|\, X_{ij}\neq 0\text{ or }X_{ji}\neq 0 \big{\}}$.

The triad-wise structural balance described in the Introduction is mathematically formalized as follows.
\smallskip
\begin{definition}[Triad-wise structural balance]\label{def:triad-wise-balance}
An appraisal network $\dsG(X)$ satisfies triad-wise structural balance, or, equivalently, is triad-wise balanced, if the appraisal matrix $X$ satisfies the following properties:
\begin{enumerate}[topsep=-7pt,label={P\arabic*:}]
\item (Symmetric appraisals) $X_{ij}X_{ji}>0$ for any $(i,j)\in \dE^+\cup \dE^-$;
\item (Positive triads) $X_{ij}X_{jk}X_{ki}>0$ for any triad $(i,j,k)$ in $\dsG(X)$.
\end{enumerate} 
\end{definition}
Note that, by property P1 above, any appraisal network satisfying triad-wise structural balance is bilateral. 

The two-faction structural balance is defined below.
\smallskip
\begin{definition}[Two-faction structural balance]\label{def:two-faction-balance}
An appraisal network $\dsG(X)$ satisfies two-faction structural balance, or, equivalently, is two-faction balanced, if either $\dsG(X)$ has no negative link or its nodes set $V$ can be partitioned into two disjoint sets $V_1$ and $V_2$ such that $X_{ij}\ge 0$, for any $i,\,j\in V_1$ or any $i,\, j\in V_2$, and $X_{ij}\le 0$, for any $i\in V_1$, $j\in V_2$, or any $i\in V_2$, $j\in V_1$.
\end{definition}

The following lemma was first proposed in~\cite{DC-FH:56} and provides a necessary and sufficient condition for two-faction structural balance in strongly connected graphs.
\smallskip
\begin{lemma}[Cycle-wise structural balance]\label{lem:cycle-to-two-faction}
Given a bilateral and strongly connected appraisal network $\dsG(X)$, it satisfies two-faction structural balance if and only if every cycle in $\dsG(X)$ is positive.
\end{lemma}

\section{Relations between Triad-Wise Structural Balance and Two-Faction Structural Balance}

\subsection{General relations}

Triad-wise structural balance is a local feature of signed appraisal networks, while two-faction structural balance characterizes some global structure. It is well known that, in all-to-all appraisal networks, these two definitions are equivalent~\cite{DC-FH:56}. In non-all-to-all bilateral appraisal networks, we have the following obvious fact.
\begin{facts}\label{fact:two-faction-implies-traid-wise}
Given a bilateral appraisal matrix $X\in \{-1,0,1\}^{n\times n}$, the associated appraisal network $\dsG(X)$ satisfies triad-wise structural balance if $\dsG(X)$ satisfies two-faction structural balance.
\end{facts}
In general, triad-wise structural balance does not imply two-faction structural balance, e.g., see Figure~\ref{fig:diff_def_struc_balance}. However, some meaningful connections can still be made from triad-wise structural balance to two-faction structural balance. For any $i\in V$ in $\dsG(X)$, define $N_i=\{j\in V\,|\,X_{ij}\neq 0\}\cup \{i\}$ and define the induced subgraph $\dsG_{N_i}(X)$ as node $i$'s \emph{ego-network}, i.e., the induced subgraph in $\dsG(X)$ involving node $i$ itself and all the nodes that node $i$ has links to. We have the following proposition and its proof is given in Appendix~\ref{appendix:prop:ego-net-balance}.

\begin{figure} 
  \captionsetup[subfigure]{}
  \centering
  \subfloat[Graph 1]{%
    \includegraphics[width=0.3\linewidth]{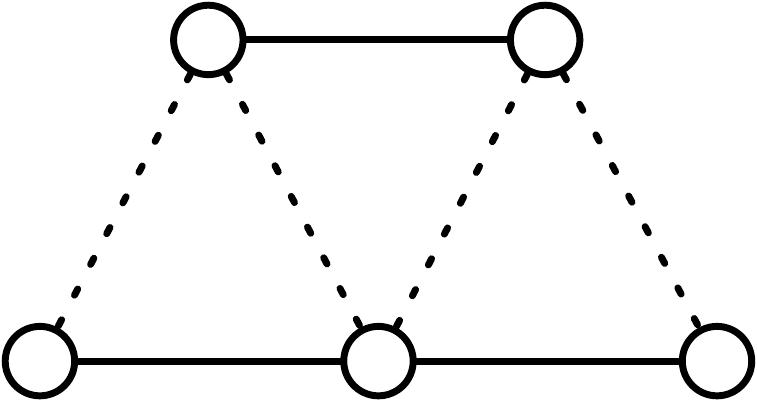}}
  \hspace{1.5cm}
  \subfloat[Graph 2]{%
    \includegraphics[width=0.33\linewidth]{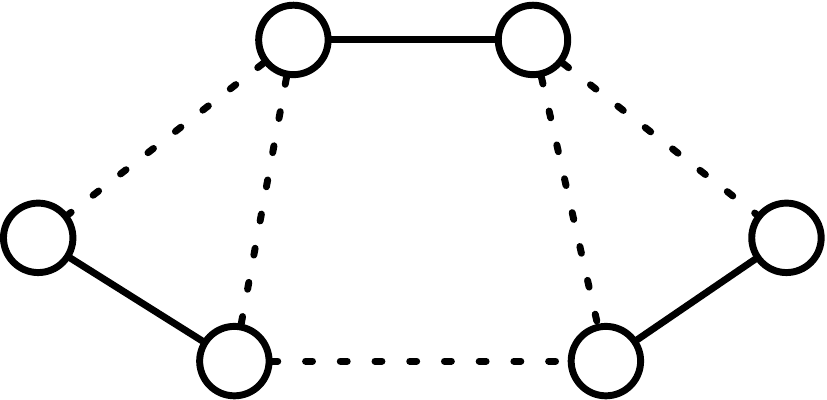}}
  \caption{Examples of the difference between triad-wise structural balance and two-faction structural balance. Here all the directed links are bilateral and sign-symmetric. Therefore the arrows are omitted. Solid lines mean positive bilateral links while dash lines mean negative bilateral links. Graph 1 satisfies both definitions of structural balance. Graph 2 satisfies triad-wise structural balance but does not satisfy two-faction structural balance. However, for each node in Graph 2, her ego-network satisfies two-faction structural balance.}
  \label{fig:diff_def_struc_balance} 
\end{figure} 

\smallskip
\begin{proposition}[Structural balance in ego-networks]\label{prop:ego-net-balance}
 For any appraisal matrix $X\in \{-1,0,1\}^{n\times n}$, if $\dsG(X)$ satisfies triad-wise structural balance, then $\dsG_{N_i}(X)$ satisfies two-faction structural balance for any $i\in V$.
\end{proposition}

Proposition~\ref{prop:ego-net-balance} has a clear sociological interpretation. Since individuals strive to resolve the cognitive dissonance generated from structural imbalanced~\cite{FH:44,FH:46} and since, intuitively, individuals can feel the imbalance from only their ego-networks, as long as the entire appraisal network satisfies triad-wise balance, all the individuals' ego-networks satisfy structural balance in both the triad-wise sense and the two-faction sense. Therefore, whenever an appraisal network satisfies triad-wise structural balance, even if it is not two-faction balanced, individuals should not feel any cognitive dissonance and thereby have no motivation to further adjust their appraisals of others. In Section 4, we will see a profound implication of Proposition~\ref{prop:ego-net-balance}: any sort of local appraisal dynamics driven by imbalance in individuals' ego-networks do not necessarily achieve the global two-faction balance.

\subsection{Conditions for equivalence}

In this subsection, we further investigate under what conditions the triad-wise structural balance is equivalent to the two-faction structural balance in non-all-to-all appraisal networks. The main results reply on some important concepts and properties presented below by Definitions~\ref{def:maximal-cyclic-subgraph},~\ref{def:chordal-graph},~\ref{def:subchordal-cycle}, Facts~\ref{fact:chord-split-cycle-into-2-cycles},~\ref{fact:triad-partition-subchordal}, and Lemma~\ref{lem:induced-graphs-chordal}. See Figure~\ref{fig:visual-concepts-chordal} for their visualized illustrations. 
\smallskip
\begin{definition}[Maximal cyclic subgraph]\label{def:maximal-cyclic-subgraph}
Consider an undirected unsigned graph $\uG=(V,E)$. An induced subgraph $\uG_{\tilde{V}}$ with $\tilde{V}\subseteq V$ is a maximal cyclic subgraph if it satisfies the following two conditions:
\begin{enumerate}[topsep=-7pt,itemsep=2pt]
\item There exists a cycle $(i_1,\dots, i_m)$ in $\uG$ passing through exactly all the nodes in $\tilde{V}$, i.e., $\tilde{V}=\{i_1,\dots, i_m\}$;
\item For any $\hat{V}$ such that $\tilde{V}\subset \hat{V}$, there does not exist a cycle in $\uG$ that passes through all the nodes in $\hat{V}$, i.e., the nodes $i_1,\dots, i_m$ are not in any cycle longer than $(i_1,\dots, i_m)$ in $\uG$.
\end{enumerate}  
\end{definition}
\smallskip
\begin{definition}[Chords and chordal graphs~\cite{JRSB-BP:93}]\label{def:chordal-graph}
In undirected unsigned graphs, a link is a chord of a cycle if it connects two non-consecutive nodes on that cycle. An undirected unsigned graph is a chordal graph if every cycle with length greater than three has at least one chord. A triad alone is also a chordal graph.
\end{definition}
Regarding cycles and chords, the following is obvious.
\begin{facts}\label{fact:chord-split-cycle-into-2-cycles}
Given a cycle $(i_1,\, i_2\, \dots,\, i_m)$ in an undirected unsigned graph $\uG=(V,E)$, if $\{i_p,i_q\}\in E$ is a chord of this cycle, with $q>p+1$, then $(i_1,\,\dots,\, i_p,\, i_q,\, i_m)$ and $(i_p,\, i_{p+1},\, \dots,\, i_q)$ are both cycles in graph $\uG$. 
\end{facts}

The following lemma about chordal graphs are frequently used in this paper.
\smallskip
\begin{lemma}[Induced chordal subgraphs~\cite{JRSB-BP:93}]\label{lem:induced-graphs-chordal}
Every induced subgraph of a chordal graph is also chordal.
\end{lemma}
\smallskip

\begin{definition}[Subchordal cycles]\label{def:subchordal-cycle}
A cycle $(i_1,\dots,i_m)$ in an undirected unsigned graph $\uG=(V,E)$ is a subchordal cycle if in $\uG$ there exists a subgraph $\uG' = (V',E')$ associated with $\C$ such that
\begin{enumerate}[topsep=-7pt,itemsep=2pt]
\item $V'=\{i_1,\dots, i_m\}$, i.e., the node set of $\uG'$ is the set of all the nodes in the cycle;
\item $\big{\{} \{i_1,i_2\},\, \{i_2,i_3\},\,\dots,\, \{i_m,i_1\} \big{\}}\subseteq E'$, i.e., $\uG'$ contains all the links the cycle;
\item $\uG'$ is a chordal graph.
\end{enumerate} 
\end{definition}
\smallskip

The concept of subchordal cycles has a very clear geometric interpretation and the proof is in Appendix~\ref{appendix:triad-partition-subchordal}.
\begin{facts}\label{fact:triad-partition-subchordal}
Suppose a cycle $(i_1,\dots,i_m)$ in an undirected unsigned graph $\uG$ is a subchordal cycle, associated with a subgraph $\uG'$ as defined in Definition~\ref{def:subchordal-cycle}. If $\uG'$ is embedded on a plane such that the cycle $(i_1,\dots, i_m)$ forms a convex polygon, then there exists a set of triads in $\uG'$, i.e., triangles on the plane, that forms a non-overlapping partition of the polygon.
\end{facts}

\begin{figure}
\centering
\includegraphics[width=0.65\linewidth]{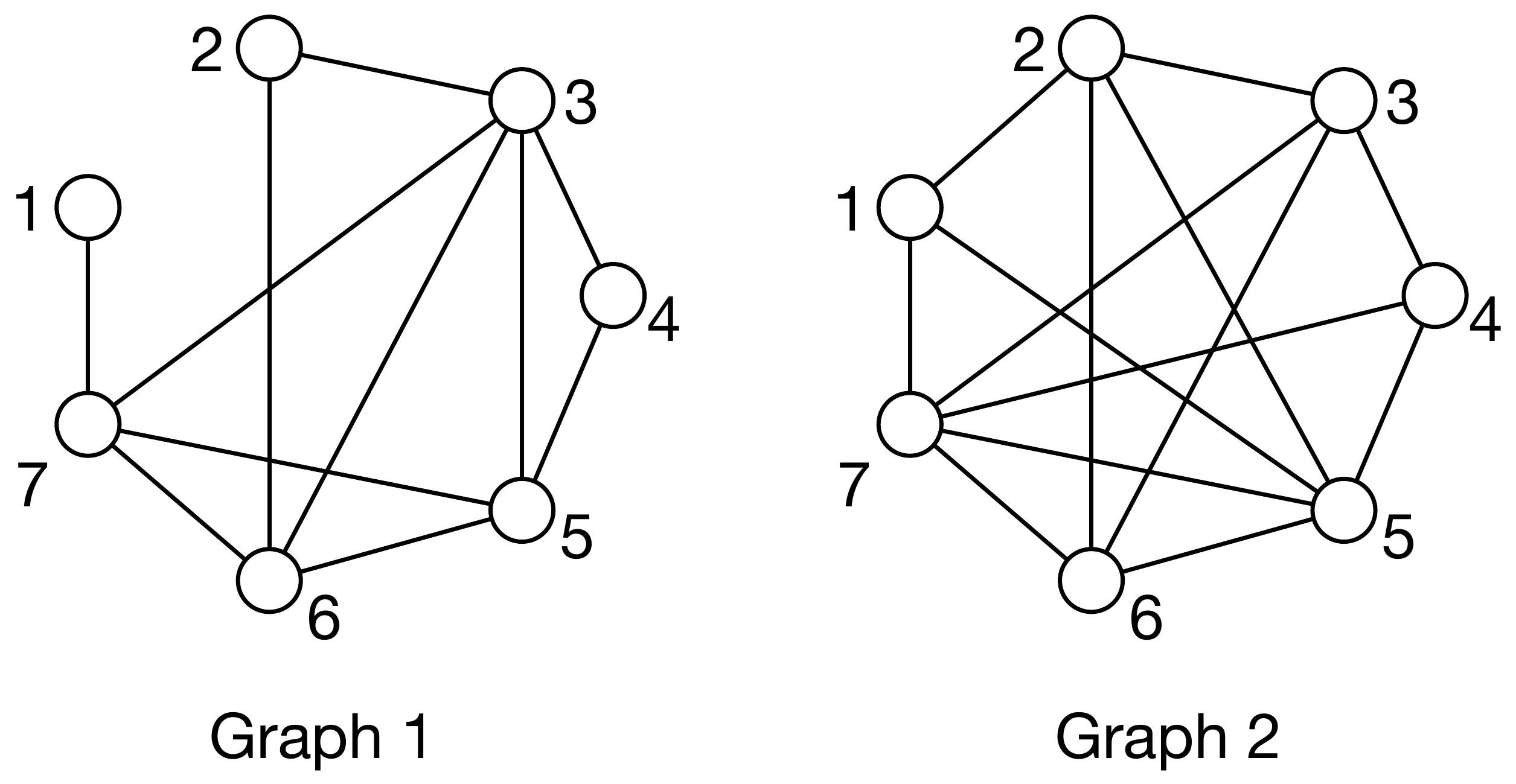}
\caption{Visualized illustrations of maximal cyclic subgraphs, chordal graphs, and subchordal cycles. Denote Graph 1 (undirected and unsigned) by $\uG$. The induced subgraph $\uG_{\{3,4,5,6,7\}}$ is a maximal cyclic subgraph, while the induced subgraph $\uG_{\{3,4,5,6\}}$ is not. The links $\{3,6\}$, $\{3,5\}$, and $\{5,7\}$ are chords of the cycle $(3,4,5,6,7)$. As indicated by Fact~\ref{fact:chord-split-cycle-into-2-cycles}, $\{3,6\}$ splits the cycle $(3,4,5,6,7)$ into two cycles $(3,6,7)$ and $(3,4,5,6)$. Graph 1 is a chordal graph according to Definition~\ref{def:chordal-graph}, and one could easily check that any of its induced subgraphs is also chordal, as indicated by Lemma~\ref{lem:induced-graphs-chordal}. Graph 2 is not a chordal graph since some of the cycles with length greater than 3, e.g., $(3,4,5,6)$, does not have a chord. However, the cycle $(3,4,5,6,7)$ in Graph 2 is a subchordal graph since the subgraph $\uG' = (V',E')$, with $V'=\{3,4,5,6,7\}$ and $E'=\big{\{} \{3,4\},\, \{4,5\},\, \{5,6\},\, \{6,7\},\, \{7,3\},\, \{4,7\},\, \{5,7\} \big{\}}$ is a chordal graph. As indicated by Fact~\ref{fact:triad-partition-subchordal}, if $(3,4,5,6,7)$ in Graph 2 is embedded on a plane as a convex polygon, then there exist triads $(3,4,7)$, $(4,5,7)$, and $(5,6,7)$ in $\uG'$ corresponding to a non-overlapping partition of this polygon.}\label{fig:visual-concepts-chordal}
\end{figure}

With all the preparation work, now we are ready to present some graph-theoretic conditions on which the triad-wise structural balance and the two-faction structural balance are equivalent. First of all, we point out that, as a straightforward consequence of statement~(iii) of Lemma~\ref{lem:properties-subchordal-cycles} and Lemma~\ref{lem:cycle-to-two-faction}, the triad-wise structural balance and the two-faction structural balance are equivalent if the corresponding undirected unsigned graph is connected and every cycle in it is subchordal. 
\smallskip
\begin{proposition}\label{prop:all-cycle-subchordal}
Given any undirected unsigned graph $\uG$ that is connected. If every cycle in it is subchordal, then, for any bilateral appraisal network $\dsG(X)$ with $\uG(X)=\uG$, $\dsG(X)$ satisfies triad-wise structural balance if and only if $\dsG(X)$ satisfies two-faction balance. 
\end{proposition}
\smallskip

The condition for the equivalence between triad-wise and two-faction balance can be further relaxed. Below we present the main theorem of this section and its proof is provided in Appendix~\ref{appendix:thm:equivalence-two-balance}.
\smallskip
\begin{theorem}[Equivalence between two definitions of structural balance]\label{thm:equivalence-two-balance}
Consider any undirected unsigned graph $\uG=(V,E)$ that is connected and satisfies the following conditions: For any maximal cyclic subgraph $\uG'$ with $m>3$ nodes,
\begin{enumerate}[topsep=-7pt,itemsep=2pt]
\item it contains a cycle $\C=(i_1,\dots, i_m)$ that passes through all the nodes in  $\uG'$ and is subchordal;
\item for any chord $\{i_p, i_q\}$ of $\C$ in $\uG$, with $q> p+1$, at least one of the two cycles $(i_1,\dots, i_p,i_q,\dots, i_m)$ and $(i_p, i_{p+1},\dots, i_q)$ in $\uG$ is subchordal.
\end{enumerate}
For any bilateral $X\in \{-1,0,1\}^{n\times n}$ with $\uG(X)=\uG$, $\dsG(X)$ satisfies triad-wise structural balance if and only if $\dsG(X)$ satisfies two-faction structural balance.
\end{theorem}

In Figure~\ref{fig:theorem-equivalence-balance} we provide examples in which the undirected unsigned graphs $\uG$ satisfy (or not resp.) the conditions in Theorem~\ref{thm:equivalence-two-balance} and, as the consequences, triad-wise balance in any $\dsG(X)$ with $\uG(X)=\uG$ always (or not resp.) always implies two-faction balance.

\begin{figure}
\begin{center}
\includegraphics[width=1\linewidth]{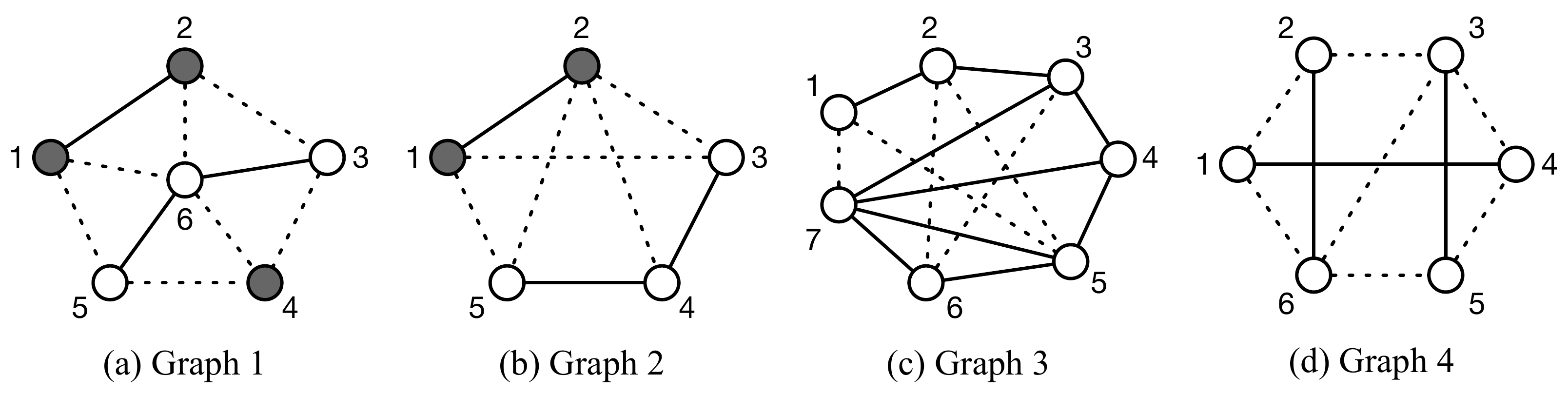}
\end{center}
\caption{Examples where triad-wise structural balance is or is not equivalent to two-faction structural balance. In all these directed signed graphs the links are bilateral and sign-symmetric, and the arrows of links are thereby omitted. Solid lines represent positive bilateral links while dashed lines represent negative bilateral links. In panel (a), the corresponding undirected unsigned graph satisfies teh conditions in Proposition~\ref{prop:all-cycle-subchordal}. In panel (b), the corresponding undirected unsigned graph satisfies all the conditions in Theorem~\ref{thm:equivalence-two-balance}, but does not satisfies the conditons in Propisition~\ref{prop:all-cycle-subchordal} since the cycle $(1,3,4,5)$ is not subchordal. Graph 1 and Graph 2 are examples where the directed signed appraisal networks satisfy both triad-wise and two-faction structural balance. The grey nodes form one faction and the other nodes form another faction. In panel (c), the corresponding undirected unsigned graph $\uG$ itself is maximal cyclic but not subchordal. As a result, Graph 3 satisfies triad-wise structural balance but its nodes can not be partitioned into two factions. In panel (d), the cycle $(1,2,\dots,6)$ in the corresponding undirected unsigned graph is subchordal but the link $\{1,4\}$ violates condition~ii) in Theorem~\ref{thm:equivalence-two-balance}. As a result, Graph 4 is triad-wise balanced but not two-faction balanced.}\label{fig:theorem-equivalence-balance}
\end{figure}

\section{Convergence of Appraisal Networks to Non-All-to-All Structural Balance}

\subsection{Convergence to triad-wise structural balance}

In this subsection we propose and analyze a discrete-time gossip-like model that characterizes how appraisal networks evolves to triad-wise structural balance with non-all-to-all topologies via local interactions. This model is built on three sociologically intuitive mechanisms: 1) the symmetry mechanism~\cite{RME:76}, i.e., individuals tend to be friendly (unfriendly resp.) to those who are friendly (unfriendly resp.) to themselves; 2) the influence mechanism~\cite{NEF-ECJ:11}, i.e., any individual $i$'s appraisal of individual $j$ is influenced by individual $i$'s friends' and enemies' appraisals of individual $j$; 3) the homophily mechanism~\cite{PFL-RKM:54}, i.e., individual $i$ tends to be friendly to individual $j$ if they have similar appraisals of others. Such a model is referred to as the \emph{symmetry-influence-homophily} (SIH) dynamics, formally defined as follows.
\smallskip
\begin{definition}[SIH dynamics]\label{def:SIHdynamics}
Given any initial appraisal matrix $X(0)\in \{-1,0,1\}^{n\times n}$, the evolution of $X(t)$ is defined by the following stochastic process: At each time step $t\in \mathbb{N}$, randomly pick a pair of nodes $i$ and $j$, with $i\neq j$, such that at least one of $X_{ij}(t)$ and $X_{ji}(t)$ is non-zero.
\begin{enumerate}[topsep=-7pt,itemsep=2pt]
\item If there does not exists any $k\in V\setminus \{i,j\}$ such that $X_{ik}(t)X_{jk}(t)\neq 0$, then update $X_{ij}(t)$ according to the symmetry mechanism, i.e.,  
\begin{align*}
X_{ij}(t+1) = X_{ji}(t);
\end{align*}
\item If there exists $k\in V\setminus \{i,j\}$ such that $X_{ik}(t)X_{jk}(t)\neq 0$, i.e., if $i$ and $j$ have a common neighbor, then randomly pick such a common neighbor $k$ and let
\begin{align*}
X_{ij}(t+1) = 
\begin{cases}
\displaystyle X_{ji}(t),\quad & \text{with probability }p_1,\\
\displaystyle X_{ik}(t)X_{kj}(t) \quad & \text{with probability }p_2,\\
\displaystyle X_{ik}(t)X_{jk}(t) \quad & \text{with probability }p_3,
\end{cases}
\end{align*} 
for some $p_1>0$, $p_2>0$, and $p_3>0$ with $p_1+p_2+p_3=1$. These three equations above correspond to the symmetry mechanism, the influence mechanism, and the homophily mechanism respectively.
\end{enumerate}
All the other links remain unchanged from $t$ to $t+1$.
\end{definition}  

\begin{facts}\label{fact:SIH-unchanged-topology}
Along the SIH dynamics, if $\dsG(X(T))$ is bilateral for some $T\in \mathbb{N}$, then $\dsG(X(t))$ is bilateral and $\uG(X(t))=\uG(X(t))$ for any $t\ge T$. 
\end{facts}
\smallskip

Below we define the equilibria of the SIH dynamics. 
\smallskip
\begin{definition}[Equilibrium]\label{def:SIH-equilibrium}
An appraisal matrix $X\in \{-1,0,1\}^{n\times n}$, or equivalently an appraisal network $\dsG(X)$, is an equilibrium of the SIH dynamics, if $\dsG(X)$ remains unchanged after any possible update described in Definition~\ref{def:SIHdynamics} of any pair of its nodes.  
\end{definition}
The proposition below characterizes the equilibrium set of the SIH dynamics. The proof is in Appendix~\ref{appendix:prop:SIH-equilibrium-set}.
\smallskip
\begin{proposition}[Equilibrium set]\label{prop:SIH-equilibrium-set}
For any appraisal matrix $X\in \{-1,0,1\}^{n\times n}$, $X$ is an equilibrium of the SIH dynamics if and only if $\dsG(X)$ is triad-wise balanced.
\end{proposition}

Given any initial condition $X(0)\in \{-1,0,1\}^{n\times n}$, the trajectory $X(t)$ along the SIH dynamics is a stochastic process, in which the randomness comes from the update sequence, including which link is updated and which update mechanism is triggered at each time step. We establish the almost-sure convergence of the SIH dynamics by showing that, for any initial condition, there exists at least one finite update sequence along which the trajectory achieves triad-wise structural balance. This argument is formalized as a lemma and is presented in Appendix~\ref{appendix:thm:convergence-triad-wise}. A similar proof strategy has been adopted in~\cite{GC-WS-WM-FB:18n}. Note that a manual update sequence is one path to almost sure convergence. However, the constructed sequence in the proof is not unique, and there may be other paths to convergence.

Now we present the main theorem on the almost-sure convergence of the SIH dynamics to triad-wise structural balance. The proof is given in Appendix~\ref{appendix:thm:convergence-triad-wise}.
\smallskip
\begin{theorem}[Convergence to triad-wise structural balance]\label{thm:convergence-to-traid-wise-balance}\label{thm:convergence-triad-wise-balance}
Consider the SIH dynamics given by Definition~\ref{def:SIHdynamics}. For any initial condition $X(0)\in\{-1,0,1\}^{n\times n}$, the trajectory $X(t)$ almost surely reaches an equilibrium, i.e., a triad-wise balanced configuration, in finite time. 
\end{theorem} 
\smallskip

\begin{remark}
The intuition behind the convergence to triad-wise is clear: The symmetry mechanism makes the links of the appraisal network bilateral and sign-symmetric, while the influence and homophily mechanism balance each triad. In addition, since the influence mechanism and the homophily mechanism are identical if the appraisal matrix is sign-symmetric, from the proof in Appendix~\ref{appendix:thm:convergence-triad-wise}, one can observe that the almost-sure finite-time convergence of the SIH dynamics to triad-wise structural balance still holds if either the influence mechanism or the homophily mechanism is removed from the dynamics. However, they cannot be both removed, since the symmetry mechanism alone does not balance sign-symmetric but negative triads.
\end{remark}

\subsection{Convergence to two-faction structural balance}

In this subsection, we discuss the convergence of appraisal networks to two-faction structural balance. Regarding the SIH dynamics, the following fact is a straightforward consequence of Theorem~\ref{thm:equivalence-two-balance} and~\ref{thm:convergence-to-traid-wise-balance}.
\smallskip
\begin{facts}\label{fact:SIH-conv-two-faction}
For any bilateral $X(0)\in \{-1,0,1\}^{n\times n}$ such that $\uG(X(0))$ satisfies the conditions for $\uG$ in Theorem~\ref{thm:equivalence-two-balance}, $\dsG(X(t))$ along the SIH dynamics almost surely reach in finite time an equilibrium satisfying two-faction structural balance.
\end{facts}

For any arbitrary $X(0)\in \{-1,0,1\}^{n\times n}$, the SIH dynamics do not always guarantee almost-sure convergence to two-faction structural balance, since in SIH dynamics the updates of appraisals are driven by only the imbalance in individuals' ego-networks, which eliminate triad-wise imbalance but not two-faction imbalance. In order to achieve almost-sure convergence to two-faction structural balance via local interactions, we need some ``additional information'' that is accessible from individuals' ego-networks and flows over the entire graph, e.g., individuals' opinions on a certain issue. According to the early works by Heider~\cite{FH:44,FH:46} and Cartwright et al.~\cite{DC-FH:56}, besides the person-person homophily, the person-opinion homophily also plays a role in shaping the interpersonal appraisals. This inspires us to consider the co-evolution between opinion dynamics and person-opinion homophily. Such a co-evolutionary model is presented as follows and referred to as the \emph{symmetry-influence-opinion-homophily} (SIOH) dynamics.
\smallskip

\begin{definition}[SIOH dynamics]\label{def:SIOH-dyn}
Denote by $y(t)\in \{-1,1\}^{n\times 1}$ the vector of the individuals' opinions on a certain issue at time $t$. Given any $X(0)\in \{-1,0,1\}^{n\times n}$ and $y(0)$, the SIOH dynamics of the appraisal matrix $X(t)$ and the opinion vector $y(t)$ are defined as follows. At any time step $t\in \mathbb{N}$, randomly pick a pair of nodes $i$ and $j$ such that $i\neq j$ and at least one of $X_{ij}(t)$ or $X_{ji}(t)$ is non-zero. Update $X_{ij}(t)$ or $y_i(t)$ according to the following rule: If $X_{ij}(t)=0$, then update $X_{ij}(t)$ via the symmetry mechanism, i.e., $X_{ij}(t+1)=X_{ji}(t)\neq 0$; If $X_{ij}(t)\neq 0$, then 
\begin{enumerate}[topsep=-7pt,itemsep=2pt]
\item With probability $q_1$, a gossip-like opinion dynamics is triggered and $y_i(t+1)=X_{ij}(t)y_j(t)$;
\item With probability $q_2$, a person-opinion homophily mechanism is triggered and $X_{ij}(t+1)=y_i(t)y_j(t)$;
\item With probability $q_3$, the original SIH dynamics as in Definition~\ref{def:SIHdynamics} are triggered. 
\end{enumerate}
Here $q_1,\, q_2,\, q_3$ are all positive and $q_1+q_2+q_3=1$. 
\end{definition}

Similar to the SIH dynamics, a pair $(X,y)$ is defined as an equilibrium of the SIOH if $X$ and $y$ remain unchanged after any possible update. The following theorem characterizes the dynamic behavior of the SIOH dynamics.
\smallskip
\begin{theorem}[Dynamic behavior of SIOH dynamics]\label{thm:SIOH-equilibrium-convergence}
Regarding the SIOH dynamics in Definition~\ref{def:SIOH-dyn}, the following statements hold:
\begin{enumerate}[topsep=-7pt,itemsep=2pt]
\item For any $X\in \{-1,0,1\}^{n\times n}$ and $y\in \{-1,1\}^{n\times 1}$, $(X,y)$ is an equilibrium of the SIOH dynamics if and only if the following two conditions hold: 1) $\dsG(X)$ is bilateral and satisfies two-faction structural balance; 2) $y_i=y_j$ if $i$ and $j$ are in the same faction, while $y_i=-y_j$ if $i$ and $j$ are in different factions;
\item For any initial appraisal matrix $X(0)\in \{-1,0,1\}^{n\times n}$ and initial opinion vector $y(0)\in \{-1,1\}^{n\times 1}$, the trajectory $(X(t),y(t))$ along the SIOH dynamics almost surely reaches an equilibrium in finite time. As a result, the appraisal network $\dsG(X(t))$ almost surely reaches a two-faction structurally balanced configuration in finite time. 
\end{enumerate}
\end{theorem}

\section{Numerical Studies}

In signed appraisal networks, since negative links represent antagonistic relations, the ratio of negative links among all the links intuitively reflects the \emph{degree of conflicts} in a social system. In this subsection, we investigate by numerical studies how degrees of conflicts in the final equilibrium states of the SIH dynamics are influenced by the initial degrees of conflicts and the structure of the appraisal networks. 

Given an initial appraisal matrix $X(0)$, the initial degree of conflicts, denoted by $c_0$, is defined as 
\begin{align*}
c_0= \sum_{i,j\in V} \vect{1}_{\{X_{ij}(0)<0\}}\,\,\Big{/}\sum_{i,j\in V} \vect{1}_{\{X_{ij}(0)\neq 0\}}.
\end{align*}
Here $\vect{1}_{\{\cdot\}}$ is the indicator function. Starting from $X(0)$, we know from Theorem~\ref{thm:convergence-to-traid-wise-balance} that the SIH dynamics almost surely reach an equilibrium at some time $T>0$. Denote by $c_{\infty}$ the final degree of conflicts, i.e., 
\begin{align*}
c_{\infty} = \sum_{i,j\in V} \vect{1}_{\{X_{ij}(T)<0\}}\,\,\Big{/}\sum_{i,j\in V} \vect{1}_{\{X_{ij}(T)\neq 0\}}.
\end{align*}
Given $X(0)$, $c_{\infty}$ is a random variable depending on the update sequence. As for the network structure, we first mainly focus on the link density. Given an appraisal network $\dsG(X)$, its link density $\rho_{\text{link}}$ is defined as 
\begin{align*}
\rho_{\text{link}}=\sum_{i,j\in V} \frac{\,\vect{1}_{\{X_{ij}\neq 0\}}\,}{n(n-1)},
\end{align*}
i.e., the number of links divided by the number of all the possible links between the nodes. In the SIH dynamics, since the zero pattern of the appraisal network remains unchanged, $\rho_{\text{link}}$ is constant once $X(0)$ is given.

We simulate the SIH dynamics on signed bilateral Erd\H{o}s-R\'{e}nyi random graphs with $n=8$ nodes. The initial appraisal networks are constructed as follows: For any pair of nodes $i,\,j$, links $(i,j)$ and $(j,i)$ are simultaneously built with probability $p$ and then, for any directed link $(i,j)$, its sign is flipped to $-1$ with some probability $p_{\text{neg}}$. Apparently, in such graphs, the expectation of link density $\rho_{\text{link}}$ is equal to $p$ and the expectation of the initial degree of conflicts is $p_{\text{neg}}$. In each simulation, once the initial condition $X(0)$ is constructed, we randomly generate valid updates until $X(t)$ reaches an equilibrium and then compute the final degree of conflicts. 

We first fix the value of $p$ and generate $X(0)$ with randomly picked values of $p_{\text{neg}}$ for 3000 times. That is, we generate 3000 different initial conditions $X(0)$ and $c_0$, from each of which we obtain the $c_{\infty}$. The scatter plots for the 3000 data pairs $(c_0,c_{\infty})$ under different fixed values of $p$ are presented in Figure~\ref{fig:final_conflict_initial_conflict}(a)-(d) respectively. Linear regressions between $c_0$ and $c_{\infty}$ are also conducted for each scatter plot. As indicated by Figure~\ref{fig:final_conflict_initial_conflict}, in general, $c_{\infty}$ has a tendency of increasing with $c_0$, which is quite intuitive, and, moreover, their correlation becomes weaker in denser networks, i.e., networks with larger $p$. Second, we fix the value of $p_{\text{neg}}$ and investigate the relation between $c_{\infty}$ and $\rho_{\text{link}}$ in a similar way based on 3000 times of independent simulations. Figure~\ref{fig:final_conflict_link_density}(a)-(d) show the simulation results on scatter plots between $c_{\infty}$ and $\rho_{\text{link}}$, as well as their corresponding linear regressions, under different fixed values of $p_{\text{neg}}$. As indicated by Figure~\ref{fig:final_conflict_link_density}(a)-(d), for networks with low initial degree of conflicts, $c_{\infty}$ increases with $\rho_{\text{link}}$, while, for networks with high initial degree of conflicts, $c_{\infty}$ decreases with $\rho_{\text{link}}$. Moreover, such correlations become less prominent when $p_{\text{neg}}$ is close to 0.5. To sum up, these aforementioned simulation results lead to clear interpretations: i) The more initial conflicts, the more final conflicts, especially in relatively sparse networks; ii) with low degree of initial conflicts, sparser networks lead to fewer final conflicts; On the contrary, with high degree of initial conflicts, denser networks lead to fewer final conflicts.

\begin{figure}
\begin{center}
\includegraphics[width=0.95\linewidth]{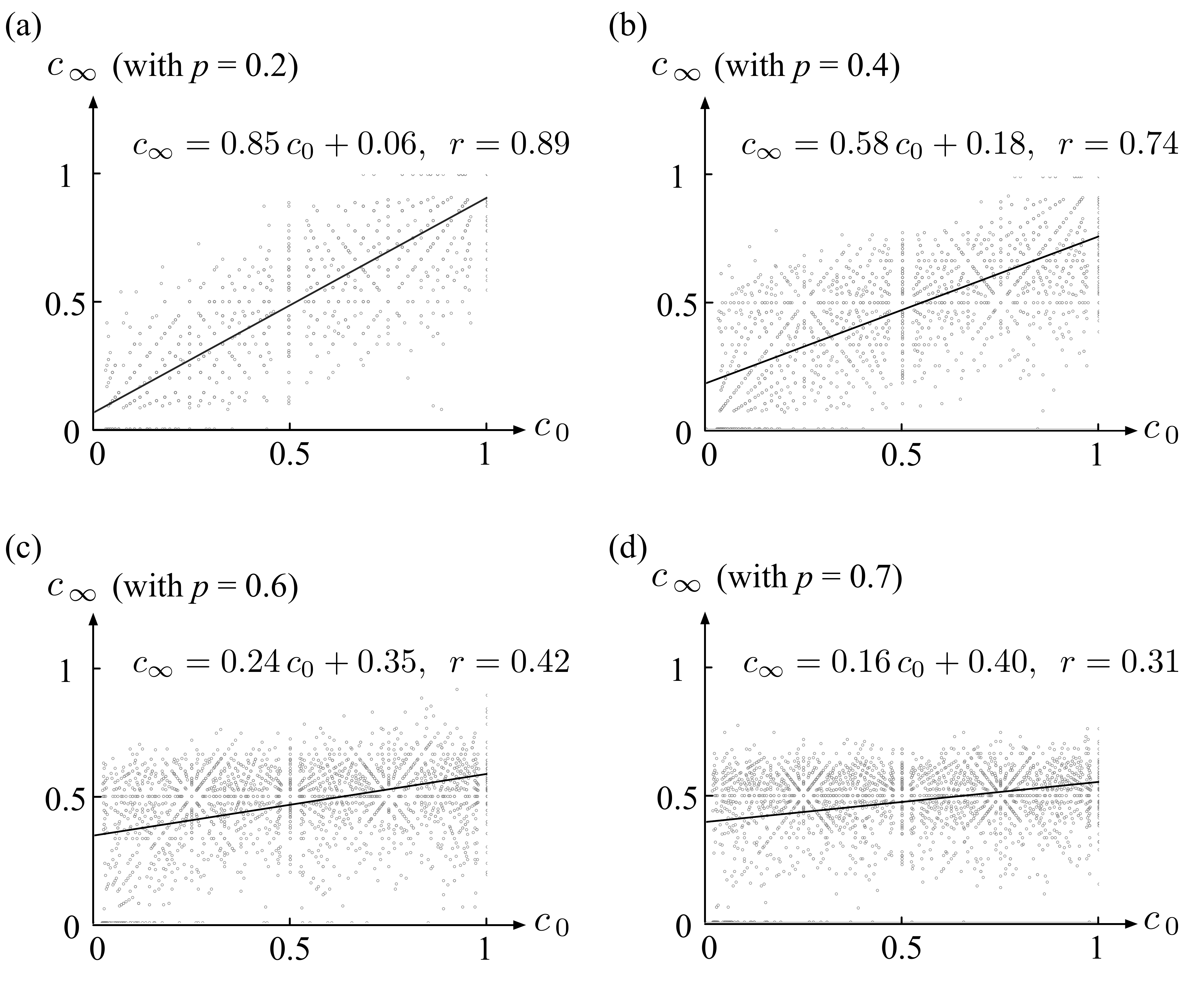}
\end{center}
\caption{Scatter plots between initial degrees of conflicts ($c_0$) and final degrees of conflicts ($c_{\infty}$) for the SIH dynamics on signed bilateral Erd\H{o}s-R\'{e}nyi graph with $p=0.2,\,0.4,\,0.6,\,0.7$ respectively. Each scatter plot contains 3000 data pairs $(c_0,c_{\infty})$. The linear equation shown in each figure is the linear regression result for each scatter plot, i.e., $c_{\infty}=k c_0+b$, and $r$ is the correlation coefficient between $c_{\infty}$ and $c_0$.}\label{fig:final_conflict_initial_conflict}
\end{figure} 

\begin{figure}
\begin{center}
\includegraphics[width=0.95\linewidth]{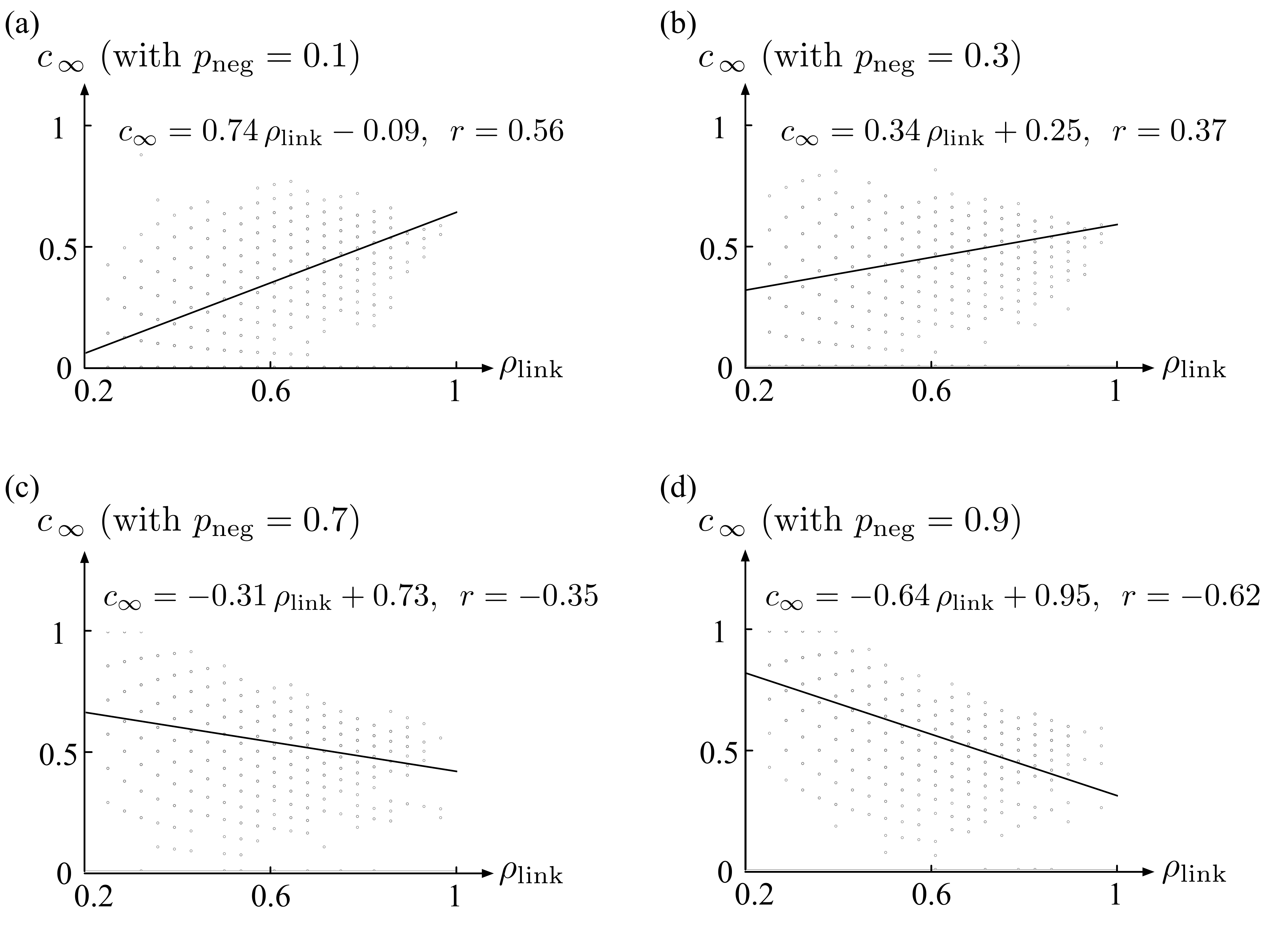}
\end{center}
\caption{Scatter plots between link densities ($\rho_{\text{link}}$) and final degrees of conflicts ($c_{\infty}$) for the SIH dynamics on signed bilateral Erd\H{o}s-R\'{e}nyi graph with $p_{\text{neg}}=0.1,\,0.3,\,0.7,\,0.9$ respectively. Each scatter plot contains 3000 data pairs $(\rho_{\text{link}},c_{\infty})$. The linear equation shown in each figure is the linear regression result for each scatter plot, i.e., $c_{\infty}=k\rho_{\text{link}}+b$, and $r$ is the correlation coefficient between $c_{\infty}$ and $\rho_{\text{link}}$.}\label{fig:final_conflict_link_density}
\end{figure} 

We further speculate that the effect of link density on final degree of conflicts might be related to the fact that denser networks contain more triads, which represent the ``multilateral relations'' between nodes in the appraisal network. Denote by $n_{\text{triad}}$ the number of triads in a given graph. In order to separate the effect of $n_{\text{triad}}$ from the effect of $\rho_{\text{link}}$, we simulate the SIH dynamics on independently randomly generated Erd\H{o}s-R\'{e}nyi with fixed values of $p$ and $p_{\text{neg}}$, and then obtain the scatter plots between $c_{\infty}$ and $n_{\text{triad}}$. As indicated by Figure~\ref{fig:final_conflict_num_triads}, the effect of the number of triads leads to very clear sociological interpretations: Multilateral relations are negatively correlated with the final degree of conflicts when the initial degree of conflicts is low, but, with high initial degrees of conflicts, multilateral relations are positively correlated with the final degree of conflicts. In addition, such correlation is more prominent in sparser networks.

\begin{figure}
\begin{center}
\includegraphics[width=0.99\linewidth]{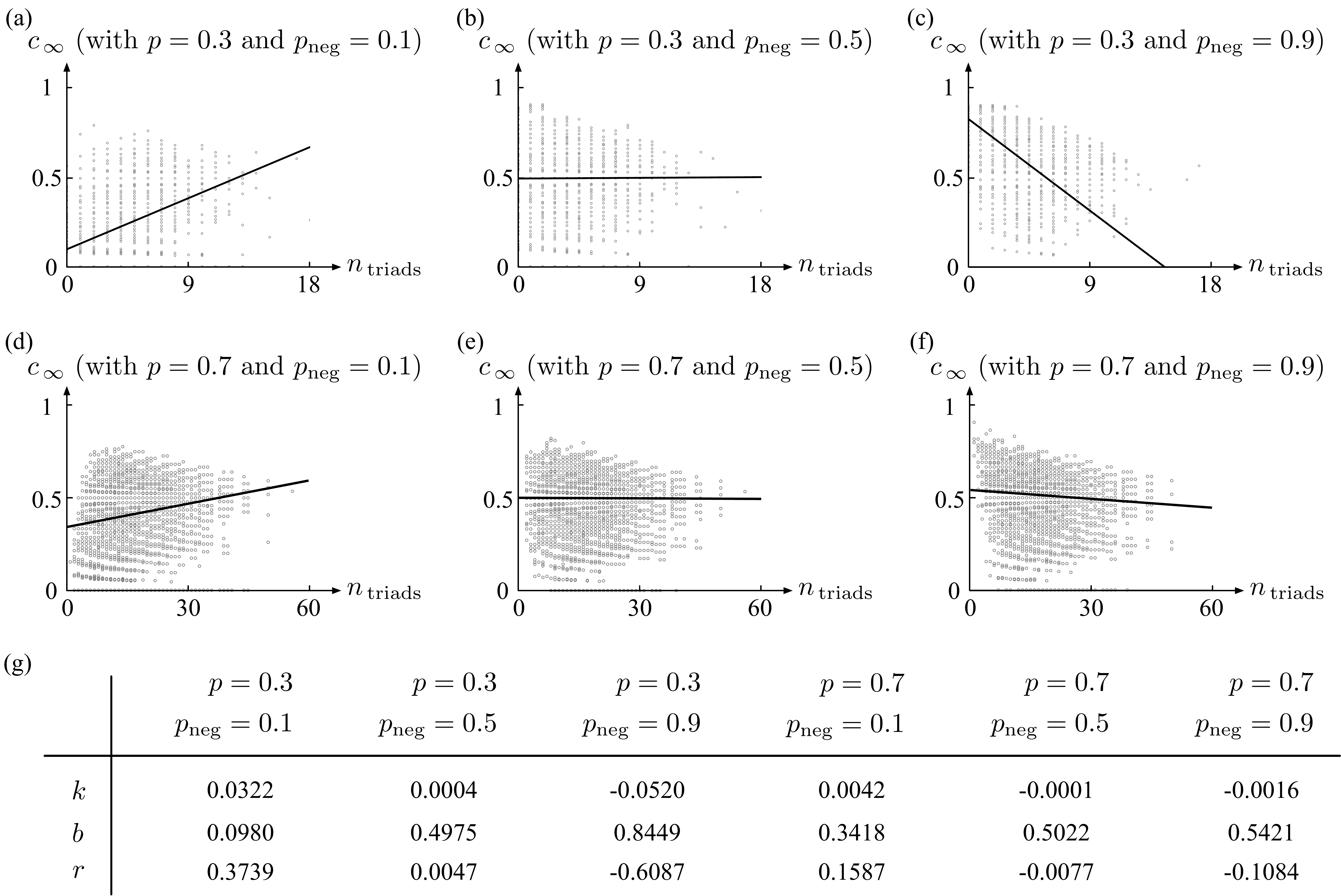}
\end{center}
\caption{Effect of the number of triads on the final degree of conflicts in the SIH dynamics. Panel~(a)-(f) are the scatter plots between the final degree of conflicts and the number of triads in the appraisal network, under different sets of values $p$ and $p_{\text{neg}}$. Each scatter plots contains 20000 data pairs $(c_{\infty},n_{\text{triads}})$ obtained by simulations on the Erd\H{o}s-R\'{e}nyi graphs independently randomly generated based on the given parameters $p$ and $p_{\text{neg}}$. Panel~(g) provides the mean-square estimations of the parameters $k$ and $b$ in the linear regressions $c_{\infty}=kn_{\text{triads}} + b$, and the correlation coefficients $r$ between $c_{\infty}$ and $n_{\text{triads}}$.}\label{fig:final_conflict_num_triads}
\end{figure}

\section{Conclusion}
This paper addresses the open problem how an interpersonal appraisal network converges to a non-all-to-all structural balance configuration. We first introduce two well-justified and intuitive definitions of non-all-to-all structural balance: the triad-wise structural balance and the two-faction structural balance, and establish the graph-theoretic conditions under which they are equivalent. We then propose two discrete-time gossip-like dynamics models of interpersonal appraisals based on symmetry, influence, and (person-to-person or person-opinion) homophily mechanisms, referred to as the SIH dynamics and the SIOH dynamics respectively. We conduct a comprehensive analysis of the dynamical behavior of these models. We prove that the set of equilibria of the SIH (SIOH resp.) dynamics is equal to the set of all the possible triad-wise (two-faction resp.) structural balance configurations. Moreover, we prove that for any initial condition, the appraisal networks in the SIH (SIOH resp.) dynamics achieves a triad-wise (two-faction resp.) structural balance configuration almost surely in finite time. Numerical studies on the SIH dynamics lead to some insightful sociological interpretations. For example, multilateral relations are negatively correlated with the final degree of conflicts when the initial degree of conflicts is low, but, with high initial degrees of conflicts, multilateral relations are positively correlated with the final degree of conflicts. In addition, such correlation is more prominent in sparser networks.

Future research directions include the investigation on the conditions for the convergence of the SIH dynamics to all-friendly appraisal networks, a special class of structural balance configurations. The necessary and sufficient graph-theoretic condition for the equivalence between triad-wise structural balance and two-faction structural balance is also an important open problem. It is also of research value to consider extending the applicability of the SIH dynamics to weighted signed graphs.
\bigskip

\begin{appendix}
\noindent\textbf{Appendix}
\vspace{-0.2cm}
\section{Proof of Proposition~\ref{prop:ego-net-balance}}\label{appendix:prop:ego-net-balance}
\emph{Proof:} Since $\dsG(X)$ is triad-wise balanced, we have $X_{ij}=X_{ji}$ for any $i,j\in V$. For any given $i\in V$, defined
\begin{align*}
V_i & = \{j\in N_i\,|\,X_{ij}=X_{ji}=1\}\cup \{i\},\\
\tilde{V}_i & = \{j\in N_i\,|\, X_{ij}=X_{ji}=-1\}.
\end{align*}
If $\tilde{V}_i=\phi$, then $V_i=N_i$. For any $j,k\in N_i\setminus \{i\}$, if $(j,k)$ is a link in $\dsG(X)$, then $(i,j,k)$ is a triad and thereby $X_{ij}X_{jk}X_{ki}=1$, which in turn leads to $X_{jk}=1$. Therefore, all the links in $i$'s ego-network are positive, i.e., $i$'s ego-network satisfies Definition~\ref{def:two-faction-balance}.

If $\tilde{V}_i\neq\phi$, due to the definition of $N_i$, we have $V_i\cap \tilde{V}_i=\phi$ and $N_i=V_i\cup \tilde{V}_i$. For any $j,k\in N_i\setminus \{i\}$, if $(j,k)\in \dE^+(X)\cup \dE^-(X)$, then $(i,j,k)$ is a triad in $\dsG(X)$ and is positive due to the triad-wise structural balance, i.e., 
\begin{align*}
X_{ij}X_{jk}X_{ki} = 1 \quad \Rightarrow \quad X_{jk}=X_{ij}X_{ki}.
\end{align*}
Therefore, by the definitions of $V_i$ and $\tilde{V}_i$,
\begin{align*}
X_{jk} = 
\begin{cases}
\displaystyle 1\quad &\text{if }j,k\in V_i,\text{ or }j,k\in \tilde{V}_i,\\
\displaystyle -1\quad & \text{if }j\in V_i,\,\,k\in \tilde{V}_i, \text{ or }j\in \tilde{V}_i,\,\, k\in V_i.
\end{cases}
\end{align*}
Moreover, since 
\begin{align*}
X_{il}=X_{li}=
\begin{cases}
\displaystyle 1\quad &\text{for any }l\in V_i,\\
\displaystyle -1\quad &\text{for any }l\in \tilde{V}_i,
\end{cases}
\end{align*}
we conclude that, in $\dsG_{N_i}(X)$, $N_i$ can be partitioned into two antagonistic factions $V_i$ and $\tilde{V}_i$, as defined in Definition~\ref{def:two-faction-balance}. Therefore, $\dsG_{N_i}(X)$ satisfies two-faction structure balance. \qed

\section{Proof of Fact~\ref{fact:triad-partition-subchordal}}\label{appendix:triad-partition-subchordal}
\emph{Proof:} We prove it by induction. Apparently, any triad, i.e., any cycle with length 3, always satisfies Fact~\ref{fact:triad-partition-subchordal}. Suppose that Fact~\ref{fact:triad-partition-subchordal} holds for any subchordal cycle with length less than $m$. Consider a subchordal cycle $\mathcal{C}=\{i_1,\dots, i_m\}$ on graph $\uG$ and an associated subgraph $\uG'$ as in Definition~\ref{def:subchordal-cycle}. Since $\uG'$ is a chordal graph, the cycle $\mathcal{C}$ on $\uG'$ has at least one chord, denoted by $\{i_j,i_k\}$ with $1\le j<k\le m$ and $j+1<k$. Due to Fact~\ref{fact:chord-split-cycle-into-2-cycles}, $\mathcal{C}_1=(i_1,\dots, i_j,i_k,\dots,i_m)$ and $\C_2=(i_j,i_{j+1},\dots, i_k)$ are both cycles on $\uG'$. That is, if we embed $\uG'$ on a plane such that $\mathcal{C}$ forms a convex polygon, then $\mathcal{C}_1$ and $\mathcal{C}_2$ also form two convex polygons and partition the polygon corresponding to $\mathcal{C}$ with no overlapping.

Moreover, since $\uG'$ is a chordal graph, according to Lemma~\ref{lem:induced-graphs-chordal} and Definition~\ref{def:subchordal-cycle}, $\mathcal{C}_1$ and $\mathcal{C}_2$ are two subchordal cycles on $\uG'$. Therefore, due to the assumption that Fact~\ref{fact:triad-partition-subchordal} holds for cycles with lengths less than $m$, there exist two sets of triads in $\uG'$ that partition the polygons corresponding to $\mathcal{C}_1$ and $\mathcal{C}_2$ with no overlapping respectively. Therefore, these two sets of triads combined together also partition the polygon corresponding to $\mathcal{C}$ with no overlapping.

\qed

\section{Proof of Theorem~\ref{thm:equivalence-two-balance}}\label{appendix:thm:equivalence-two-balance}

Before proving Theorem~\ref{thm:equivalence-two-balance}, we first present and prove the following lemma on some important properties of subchordal cycles.
\smallskip
\begin{lemma}[Properties of subchordal cycles]\label{lem:properties-subchordal-cycles}
Consider an undirected unsigned graph $\uG$ and a subchordal cycle $\mathcal{C}=(i_1,\dots, i_m)$ in it. Let $\uG'=(V',E')$ be a chordal subgraph associated with the cycle $\mathcal{C}$ as in Definition~\ref{def:subchordal-cycle}. The following statements hold:
\begin{enumerate}[topsep=-7pt,itemsep=2pt]
\item If $\{i_p,i_q\}$, with $1\le p<p+1<q\le n$, is a chord of $\C$ in $\uG'$, then the two cycles $(i_1,\dots, i_p,i_q,\dots, i_m)$ and $(i_p,i_{p+1},\dots, i_q)$ are both subchordal cycles in $\uG'$;
\item There exists at least one triad in $\uG'$ whose nodes are three consecutive nodes in $\C$. That is, there exists $k\in \{1,\dots, m\}$ such that $(i_{k-1},i_k,i_{k+1})$ is a triad in $\uG'$ (Here we take $i_0$ as $i_m$ and $i_{m+1}$ as $i_1$);
\item For any bilateral appraisal matrix $X\in \{-1,0,1\}^{n\times n}$ such that $\uG(X) = \uG$, if $\dsG(X)$ satisfies triad-wise structural balance, then the cycle $\C=(i_1,\dots, i_m)$ on $\dsG(X)$ is a positive cycle.
\end{enumerate}
\end{lemma}

\begin{proof}
Statement~(i) is a straightforward consequence of Definition~\ref{def:subchordal-cycle}, Fact~\ref{fact:chord-split-cycle-into-2-cycles}, and Lemma~\ref{lem:induced-graphs-chordal}.

Now we prove statement~(ii) by construction. Since $\uG'$ is a chordal graph, $\C$ must have a chord in $\uG'$, denoted by $\{i_{p_1},i_{q_1}\}$, with 
\begin{align*}
1\le p_1<q_1\le m\quad \text{and}\quad 1<q_1-p_1.
\end{align*}
If $q_1 = p_1+2$, then we already find the triad $(i_{p_1},i_{p_1+1},i_{q_1})$ in $\uG'$, whose nodes are three consecutive nodes in $\C$. Suppose $q_1>p_1+2$. According to Fact~\ref{fact:chord-split-cycle-into-2-cycles} and Lemma~\ref{lem:induced-graphs-chordal}, the cycle $\C_1=(i_{p_1},i_{p_1+1},\dots, i_{q_1})$ is a subchordal cycle in $\uG'$ and the induced subgraph $\uG'_{\{i_{p_1},i_{p_1+1},\dots, i_{q_1}\}}$ is a chordal graph. As a result, $\C_1$ has a chord in $\uG'_{\{i_{p_1},i_{p_1+1},\dots, i_{q_1}\}}$ (and thus also in $\uG'$), denoted by $\{i_{p_2},i_{q_2}\}$, with
\begin{align*}
p_1\le p_2<q_2\le q_1\quad \text{and}\quad 1<q_2-p_2 \le q_1-p_1-1.
\end{align*}
Iterate this argument as far as possible. For any $j$-th iteration, $(i_{p_j},i_{p_j+1},\dots, i_{q_j})$ is a subchordal cycle on $\uG'$ with
\begin{align*}
p_{j-1}\le p_j<q_j\le q_{j-1},\,\,\,\, 1<q_j-p_j \le q_{j-1}-p_{j-1}-1.
\end{align*}
Therefore, eventually we will obtain some $j^*\ge 1$ such that $(i_{p_{j^*}},i_{p_{j^*}+1},\dots,i_{q_{j^*}})$ is a subchordal cycle in $\uG'$ with $1\le p_{j^*}<q_{j^*}\le m$ and $q_{j^*}=p_{j^*}+2$. Let $k=p_{j^*}+1$ and thereby $(i_{k-1},i_k,i_{k+1})$ is a triad in $\uG'$.   

We prove statement~(iii) by induction. Apparently, if $\C$ has length 3, i.e., if $\C$ is a triad, then ``$\dsG(X)$ is triad-wise structurally balanced'' implies that $\C$ in $\dsG(X)$ is is a positive cycle. Suppose statement~(iii) holds for subchordal cycles with length less than $m$. Now we prove that it also hold for the cycle $\C$ with length $m$. According to statement~(iii), there exists a triad $(i_{k-1},i_k,i_{k+1})$ in $\uG'$ (and thus in $\uG$). Without loss of generality, assume $k>1$. Therefore, $\{i_{k-1},i_{k+1}\}$ is a chord of the cycle $\C$ in $\uG'$. According to statement~(i), the cycle $(i_1,\dots, i_{k-1},i_{k+1},\dots, i_m)$ is also a subchordal cycle. Therefore, according to the pre-assumption that statement~(iii) holds for any subchordal cycle with length less than $m$, the cycle $(i_1,\dots, i_{k-1},i_{k+1},\dots, i_m)$ in $\dsG(X)$ is a positive cycle, i.e., 
\begin{equation}\label{eq:proof-subchordal-cycle-positive1}
X_{i_1 i_2}\cdots X_{i_{k-1} i_{k+1}} \cdots X_{i_{m-1} i_m} = 1.
\end{equation}
Moreover, since $\dsG(X)$ satisfies triad-wise structural balance, we have $X_{i_{k-1} i_k}X_{i_k i_{k+1}}X_{i_{k+1} i_{k-1}}=1$ and $X_{i_{k+1} i_{k-1}}=X_{i_{k-1} i_{k+1}}$. Therefore,
\begin{equation}\label{eq:proof-subchordal-cycle-positive2}
X_{i_{k-1} i_k}X_{i_k i_{k+1}} X_{i_{k-1} i_{k+1}}=1.
\end{equation}
Multiplying both sides of equation~\eqref{eq:proof-subchordal-cycle-positive1} and~\eqref{eq:proof-subchordal-cycle-positive2}, we obtain
\begin{align*}
X_{i_1 i_2}\cdots X_{i_{k-1} i_k}X_{i_k i_{k+1}}\cdots X_{i_{m-1} i_m} ( X_{i_{k-1} i_{k+1}} )^2=1,
\end{align*}
which implies that $X_{i_1 i_2}\dots X_{i_{m-1} i_m}=1$, i.e., $\C$ is a positive cycle in $\dsG(X)$. This concludes the proof. 
\end{proof}

With Lemma~\ref{lem:properties-subchordal-cycles}, now we can prove Theorem~\ref{thm:equivalence-two-balance}.

\emph{Proof of Theorem~\ref{thm:equivalence-two-balance}:} The ``if'' part is exactly Fact~\ref{fact:two-faction-implies-traid-wise}. Here we prove the ``only if'' part. We prove the two-faction balance of $\dsG(X)$ by showing that all the cycles in $\dsG(X)$ are positive. For any cycle $\C=(i_1,\dots, i_m)$ in $\dsG(X)$, it falls into either of the following two cases.

Case 1: The induced subgraph $\uG_{\{i_1,\dots, i_m\}}$ is already a maximal cyclic subgraph. Since this cycle in $\uG$ is subchordal and $\dsG(X)$ is triad-wise balanced, according to Lemma~\ref{lem:properties-subchordal-cycles}(iii), the cycle $\C$ in $\dsG(X)$ is positive.

Case 2: $\uG_{\{i_1,\dots, i_m\}}$ is not a maximal cyclic subgraph. According to Definition~\ref{def:maximal-cyclic-subgraph}, this implies that there exists $V'$ with $\{i_1,\dots, i_m\}\subset V'$ such that the induced subgraph $\uG_{V'}$ is a maximal cyclic subgraph. According to condition~(i) of Theorem~\ref{thm:equivalence-two-balance}, there exists a cycle $\C'=(j_1,\dots, j_s)$ in $\uG$ such that $\{i_1,\dots, i_m\}\subset V'=\{j_1,\dots, j_s\}$ and $C'$ is subchordal. In addition, for any $r\in \{1,\dots,m\}$, there exists one unique $k_r\in \{1,\dots, s\}$ such that $i_r=j_{k_r}$. Note that we do not necessarily have $k_1\le k_2\le \dots \le k_m$.

Since the cycle $\C'$ is subchordal and $\dsG(X)$ is triad-wise balanced, according to Lemma~\ref{lem:properties-subchordal-cycles}(iii), $\C'$ is positive in $\dsG(X)$. Let $y_{j_1}=1$ and, for any $k\in \{2,\dots, s\}$, let
\begin{align*}
y_{j_k}=
\begin{cases}
\displaystyle 1\quad & \text{if }X_{j_1 j_2}\cdots X_{j_{k-1}j_k}=1,\\
\displaystyle -1\quad & \text{if }X_{j_1 j_2}\cdots X_{j_{k-1}j_k}=-1.
\end{cases}
\end{align*} 
In addition, since $(j_1,\dots,j_s)$ in $\dsG(X)$ is positive and $\dsG(X)$ is triad-wise balanced, we have $X_{j_s j_1}=X_{j_1 j_s}=y_{j_1}y_{j_s}$. Therefore, for any $k\in \{1,\dots,s\}$,
\begin{align*}
X_{j_k j_{k+1}} =X_{ j_{k+1} j_k}= y_{j_k}y_{j_{k+1}}.\quad \text{(We take }j_{s+1}\text{ as }j_1\text{)}.
\end{align*}

Now we look back to the cycle $(i_1,\dots, i_m)$. For any $r\in \{1,\dots, m\}$, if $|k_r-k_{r+1}|=1$, i.e., if $i_r$ and $i_{r+1}$ are also two consecutive nodes in the cycle $(j_1,\dots, j_s)$, then, according to the argument in the above paragraph,  we already have $X_{i_r i_{r+1}}=y_{i_r i_{r+1}}$. (Here we take $i_{m+1}$ as $i_1$.) If $|k_r-k_{r+1}|\ge 2$, without loss of generality, suppose $k_{r+1}\ge k_r +2$. According to condition~(ii) of Theorem~\ref{thm:equivalence-two-balance}, either $(j_1,\dots, j_{k_r}j_{k_{r+1}}\dots, j_s)$ or $(j_{k_r},j_{k_r+1},\dots, j_{k_{r+1}})$ is subchordal in $\uG$. Suppose $(j_1,\dots, j_{k_r}j_{k_{r+1}},\dots, j_s)$ is subchordal, since $\dsG(X)$ is triad-wise balanced, according to Lemma~\ref{lem:properties-subchordal-cycles}(iii), $(j_1,\dots, j_{k_r}j_{k_{r+1}},\dots, j_s)$ is positive in $\dsG(X)$. That is, 
\begin{align*}
& X_{j_1 j_2}\cdots X_{j_{k_r-1}j_{k_r}}X_{j_{k_r}j_{k_{r+1}}}X_{j_{k_{r+1}}j_{k_{r+1}+1}}\cdots X_{j_s j_1}=1\\
& \Leftrightarrow y_{j_1}^2\cdots y_{j_{k_r-1}}^2 y_{j_{k_r}}X_{j_{k_r} j_{k_{r+1}}}y_{j_{k_{r+1}}}y_{j_{k_{r+1}+1}}^2\cdots y_{j_s}^2 =1 \\
& \Leftrightarrow X_{j_{k_r}j_{k_{r+1}}} = X_{j_{k_{r+1}}j_{k_r}} = y_{j_{k_r}}y_{j_{k_{r+1}}}.
\end{align*}
Suppose $(j_{k_r},j_{k_{r}+1},\dots, j_{k_{r+1}})$ is subchordal in $\uG$, similarly, we have 
\begin{align*}
& y_{j_{k_r}}y_{j_{k_r+1}}^2\cdots y_{j_{k_{r+1}-1}}^2 y_{j_{k_{r+1}}} X_{j_{k_{r+1}}j_{k_r}}=1\\
\Leftrightarrow & X_{j_{k_{r+1}}j_{k_r}} = X_{j_{k_r} j_{k_{r+1}}} = y_{j_{k_r}} y_{j_{k_{r+1}}}.
\end{align*}
Till now we have proved that, for any $r\in \{1,\dots, m\}$, $X_{i_r i_{r+1}}=y_{i_r}y_{i_{r+1}}$. (Again we take $i_{m+1}=i_1$.) As a consequence,
\begin{align*}
X_{i_1 i_2}\cdots X_{i_m i_1} = y_{i_1}^2 \cdots y_{i_m}^2 = 1.
\end{align*}
That is, $(i_1,\dots, i_m)$ is a positive cycle in $\dsG(X)$.

Case 1 and Case 2 together imply that any cycle in \\$\dsG(X)$ is positive. Since $\dsG(X)$ is bilateral and $\uG(X)=\uG$ is connected, we have that $\dsG(X)$ is strongly connected. According to Lemma~\ref{lem:cycle-to-two-faction}, $\dsG(X)$ satisfies two-faction structural balance. This concludes the proof. \qed

\section{Proof of Proposition~\ref{prop:SIH-equilibrium-set}}\label{appendix:prop:SIH-equilibrium-set}
 \emph{Proof: } By definition, $X\!\in\! \{-1,0,1\}^{n\!\times\! n}$ is an equilibrium if and only if the following three conditions hold: i) $X_{ij}=X_{ji}$ for any $i\neq j$. That is,  $X$ is invariant under any update via symmetry mechanism; ii) $X_{ij}=X_{ik}X_{kj}$ for any $i\neq j$ and any $k$ (if any) such that $X_{ik}X_{jk}\neq 0$. That is, $X$ is invariant under any update via influence mechanism; iii) $X_{ij}=X_{ik}X_{jk}$ for any $i\neq j$ and any $k$ (if any) such that $X_{ik}X_{jk}\neq 0$. That is, $X$ is invariant under any update via homophily mechanism. Apparently, conditions i)-iii) hold if and only if $\dsG(X)$ is bilateral and $X_{ij}X_{jk}X_{ki}=1$ for any triad $(i,j,k)$ in $\dsG$, i.e, $\dsG(X)$ is triad-wise balanced. This concludes the proof.\qed

\section{Proof of Theorem~\ref{thm:convergence-to-traid-wise-balance}}\label{appendix:thm:convergence-triad-wise}

To prove Theorem~\ref{thm:convergence-to-traid-wise-balance} we need the following lemma.
\begin{lemma}[Convergence by manually constructing update sequences]\label{lem:randomness-to-control}
Given any initial condition $X(0)\in \{-1,0,1\}^{n\times n}$, let 
\begin{align*}
\Omega = \Big{\{} & Y\in \{-1,0,1\}^{n\times n} \,\Big|\, Y_{ij}\neq 0 \text{ if }X_{ij}(0)X_{ji}(0)\neq 0,\\
&\text{and }Y_{ij} = 0\text{ if }X_{ij}(0)=X_{ji}(0)=0,\, \forall i\,,j  \Big{\}}.
\end{align*}
The set $\Omega$ is invariant along the SIH dynamics. Moreover, the trajectory $X(t)$ starting with $X(0)$ almost surely reaches an equilibrium in finite time as long as, for any $\hat{X}(0)\in \Omega$, there exists an update sequence, along which the trajectory $\hat{X}(t)$ reaches an equilibrium in finite time.
\end{lemma}

\begin{proof}
For any $X(0)\!\in\! \{-1,0,1\}^{n\!\times\! n}$, based on Definition~\ref{def:SIHdynamics}, one could eaily check that $X(t)\in \Omega$ for any $t$, i.e., $\Omega$ is invariant along the SIH dynamics. Since $\Omega$ is finite, $X(t)$ is a finite-state Markov chain over the states set $\Omega$, with the absorbing states equivalent to the equilibria of the SIH dynamics. Moreover, starting from any initial condition in $\Omega$, there exists at least one update sequence along which the trajectory of the SIH dynamics reaches an equilibrium in finite time. This means that, for such a Markov chain, from any state there is a path to an absorbing state. Therefore, this Markov chain is an absorbing Markov chain. According to Theorem~11.3 in~\cite{CMG-JLS:97}, $X(t)$ with $X(0)=X_0$ almost surely converges to an equilibrium of the SIH dynamics. Moreover, since $\Omega$ is finite, $X(t)$ with $X(0)=X_0$ almost surely reaches an equilibrium in finite time.
\end{proof}

\emph{Proof of Theorem~\ref{thm:convergence-to-traid-wise-balance}: }For any given $X(0)\in \{-1,0,1\}^{n\times n}$, define $\Omega$ as in Lemma~\ref{lem:randomness-to-control}. In order to prove that $X(t)$ reaches an equilibrium of the SIH dynamics in finite time, it is sufficient to show that, for any $\tilde{X}(0)\in \Omega$, we can manually construct an update sequence, along which the trajectory $\tilde{X}(t)$ reaches an equilibrium at some finite time $\tilde{T}$. Such an update sequence is constructed as follows.

Firstly, for any $t\in \mathbb{N}$, if there exists a pair of nodes $i$ and $j$ such that one of $\tilde{X}_{ij}(t)$ and $\tilde{X}_{ji}(t)$ is non-zero and the other is zero, pick such $i$ and $j$ and update $\tilde{X}_{ij}(t)$ according to the symmetry mechanism, i.e., $\tilde{X}_{ij}(t+1)=\tilde{X}_{ji}(t)$. Along this process, the number of such pairs of nodes in the appraisal network $\dsG(X(t))$ is strictly decreasing. Therefore, there exists some time $T_1\ge 0$ such that $\tilde{X}(T_1)$ is bilateral. According to the SIH dynamics as in Definition~\ref{def:SIHdynamics}, for any $t\ge T_1$, $\tilde{X}(t)$ will always be bilateral.

Secondly, for any $t>T_1$,
\begin{enumerate}[topsep=-7pt,itemsep=2pt]
\item if there exists $i,j$ such that $\tilde{X}_{ij}(t)=-1$ and $\tilde{X}_{ji}(t)=1$, then update the link $(i,j)$ via the symmetry mechanism, i.e., $X_{ij}(t+1)=X_{ji}(t)=1$;
\item if $\tilde{X}(t)=\tilde{X}(t)^{\top}$ but $\dsG(\tilde{X}(t))$ is not triad-wise balanced, then, according to Definition~\ref{def:triad-wise-balance}, there exists a triad $(i,j,k)$ in $\dsG(\tilde{X}(t))$ such that 
\begin{align*}
& \tilde{X}_{ij}(t)=\tilde{X}_{ji}(t)=-1\quad \text{and}\\
& \tilde{X}_{ik}(t)\tilde{X}_{kj}(t)=\tilde{X}_{ik}(t)\tilde{X}_{jk}(t)=1.
\end{align*}
In this case, update the link $(i,j)$ via either the influence mechanism or the homophily mechanism by choosing $k$ as their common neighbor. As the result, we obtain 
\begin{align*}
X_{ij}(t+1)=X_{ik}(t)X_{kj}(t) = X_{ik}(t)X_{jk}(t) = 1.
\end{align*}
\end{enumerate}
By definition, if neither of the two scenarios above applies at some time $\tilde{T}$, then $\dsG(\tilde{X}(\tilde{T}))$ already satisfies triad-wise structural balance, i.e., $\tilde{X}(\tilde{T})$ is already an equilibrium of the SIH dynamics, and the update sequence terminates at $\tilde{T}$.

Define $h(X)=\sum_{i,j}\vect{1}_{\{X_{ij}<0\}}$, i.e., the number of negative links in $\dsG(X)$. Starting from any $\tilde{X}(0)\in \Omega$ and along the update sequence constructed above, we know that, as long as $\dsG(\tilde{X}(t))$ has not achieved triad-wise structural balance, from $t$ to $t+1$ one negative link in $\dsG(\tilde{X}(t))$ will be updated and flipped to positive. Therefore, $h(\tilde{X}(t))$ is monotonically decreasing before $\dsG(\tilde{X}(t))$ achieves structural balance. Since $0\le h(\tilde{X}(0))<n(n-1)$, where $n(n-1)$ is the maximum number of directed links in a graph with $n$ nodes (not including self loops), we know that the update sequence constructed in last paragraph terminates in finite time steps. Moreover, since this algorithm terminates only when the appraisal network achieves triad-wise structural balance, we conclude that $\dsG(\tilde{X}(t))$ achieves triad-wise structural balance at time $\tilde{T}<n(n-1)$. Therefore, according to Lemma~\ref{lem:randomness-to-control}, any trajectory $X(t)$ starting from $X(0)$ almost surely reaches to an equilibrium in finite time.  \qed

\section{Proof of Theorem~\ref{thm:SIOH-equilibrium-convergence}}
\emph{Proof:} We first prove the ``if'' part of statement~(i). Since $\dsG(X)$ is bilateral and satisfies two-faction balance, according to Fact~\ref{fact:two-faction-implies-traid-wise}, $\dsG(X)$ also satisfies triad-wise balance. As Proposition~\ref{prop:SIH-equilibrium-set} implies, $X$ remain unchanged after an update of any pair of its nodes via either symmetry or influence or person-person homophily mechanisms. Moreover, for any link $(i,j)$, since $y_i=y_j$ of $i,j$ are in the same faction and $y_i=-y_j$ if $i,j$ are in different factions, we have that $X_{ij}=y_iy_j$, i.e., $X_{ij}$ remains unchanged via the person-opinion homophily, and, moreover, $y_j = X_{ij}y_i$, i.e., $y_j$ does not change via the gossip-like opinion dynamics. Therefore, $(X,y)$ remains unchanged via any possible update and is thus an equilibrium of the SIOH dynamics.

Now we prove the ``only if'' part of statement~(i). If $(X,y)$ is an equilibrium of the SIOH dynamics. Firstly, since $X$ remains unchanged via any possible update via the symmetry mechanism, we have $X=X^{\top}$, which implies that $\dsG(X)$ is bilateral. Let $V_1 = \{i\,|\,y_i=1\}$ and $V_2 = \{i\,|\, y_i=-1\}$. For any pair of nodes $i$ and $j$ in $V_1$, either $X_{ij}=0$ or $X_{ij}$ remain unchanged via the person-opinion homophily mechanism, which implies that $X_{ij}=y_iy_j = 1$. Due to the same argument, for any $i,\, j\in V_2$, either $X_{ij}=0$ or $X_{ij}=1$. Similarly, for any $i\in V_1$ and $j\in V_2$, either $X_{ij}=0$ or $X_{ij}=-1$. Therefore, $\dsG(X)$ satisfies the two-faction structural balance. By the construction of $V_1$ and $V_2$, we also have $y_1=y_j$ is $i,\, j\in V_1$ or $i,\, j\in V_2$, and $y_i=-y_j$ if $i\in V_1$ and $j\in V_2$. This concludes the proof for statement~(i).

For statement~(ii), we prove it by adopting the same strategy as how we prove Theorem~\ref{thm:convergence-to-traid-wise-balance}. That is, we show that, for any initial condition $X(0)$, we can manually construct one update sequence, along which the trajectory reaches an equilibrium in finite time. For any given $X(0)\in \{-1,0,1\}^{n\times n}$, the update sequence is constructed in the following way:  

Firstly, for any $t\in \mathbb{N}$, if there exists a pair of nodes $i$ and $j$ such that one of $\tilde{X}_{ij}(t)$ and $\tilde{X}_{ji}(t)$ is non-zero and the other is zero, pick such $i$ and $j$ and update $\tilde{X}_{ij}(t)$ according to the symmetry mechanism, i.e., $\tilde{X}_{ij}(t+1)=\tilde{X}_{ji}(t)$. Along this process, the number of such pairs of nodes in the appraisal network $\dsG(X(t))$ is strictly decreasing. Therefore, there exists some time $T_1\ge 0$ such that $\tilde{X}(T_1)$ is bilateral. According to the SIOH dynamics as in Definition~\ref{def:SIOH-dyn}, for any $t\ge T_1$, $\tilde{X}(t)$ will always be bilateral.

Secondly, for any $t>T_1$,
\begin{enumerate}[topsep=-7pt,itemsep=2pt]
\item if there exists $i,j$ such that $\tilde{X}_{ij}(t)=-1$ and $\tilde{X}_{ji}(t)=1$, then update the link $(i,j)$ via the symmetry mechanism, i.e., $X_{ij}(t+1)=X_{ji}(t)=1$;
\item if $X(t)$ is sign-symmetric but there exists a link $(i,j)$ such that $X_{ij}(t)=-1$ and $y_i(t)y_j(t)=1$, then update $X_{ij}(t)$ via the person-opinion homophily;
\item if $X(t)$ is sign-symmetric but there exists a link $(i,j)$ such that $X_{ij}(t)=1$, $y_i(t)=-1$, and $y_j(t)=1$, then update $y_i(t)$ via the gossip-like opinion dynamics.
\end{enumerate}
If, at some $T>T_1$, none of the above applies, then $X(T)$ is sign-symmetric and $X_{ij}(T)=y_i(T)y_j(T)$ for any link $(i,j)$. That is, $\dsG(X(T))$ is bilateral and already achieves two-faction structural balance. The two factions are $V_1=\{i\,|\, y_i(T)=1\}$ and $V_2 = \{i\,|\, y_i(T)=-1\}$, which in turn implies that $(X(t),y(t))$ reaches an equilibrium of the SIOH dynamics.

Define $h(X,y)=\sum_{i,j}\vect{1}_{\{X_{ij}<0\}}+\sum_i \vect{1}_{\{y_i<0\}}$. For $t\ge T_1$, $h(X(t),y(t))$ is strictly decreasing if at least one of cases (i)-(iii) occurs. Since $h(X,y)\ge 0$ for any $X$ and $y$, there must be some $T\ge T_1$ such that none of cases (i)-(iii) occurs, which implies that $(X(t),y(t))$ reaches an equilibrium at the finite time step $T$. This concludes the proof. 
\qed

\end{appendix}

\bibliographystyle{plainurl}
\bibliography{alias,Main,FB,new}

\end{document}